\documentclass[amsmath,amsfonts,aps,onecolumn,twoside,superscriptaddress,12pt]{revtex4}

\usepackage{comment}

\usepackage{amsfonts}
\usepackage{amsmath}
\usepackage{graphics}
\usepackage{graphicx}
\usepackage{enumerate}
\usepackage{amssymb}
\usepackage{amsthm}
\usepackage{subfigure}

\newcommand{\appref}[1]{\hyperref[#1]{{Appendix~\ref*{#1}}}}
\newcommand{\be}{\begin{eqnarray} \begin{aligned}}
\newcommand{\ee}{\end{aligned} \end{eqnarray} }
\newcommand{\benn}{\begin{eqnarray*} \begin{aligned}}
\newcommand{\eenn}{\end{aligned} \end{eqnarray*}}

\newcommand*{\cE}{\mathcal{E}}
\newcommand*{\cF}{\mathcal{F}}

\newcommand*{\cM}{\mathcal{M}}
\newcommand*{\cN}{\mathcal{N}}
\newcommand*{\cD}{\mathcal{D}}

\newcommand*{\cP}{\mathcal{P}}

\newcommand{\bc}{\begin{center}}
\newcommand{\ec}{\end{center}}

\newcommand{\dg}{\dagger}

\newcommand{\Tr}{\mathop{\mathrm{tr}}\nolimits}

\newtheorem{theorem}{Theorem}[section]
\newtheorem{lemma}[theorem]{Lemma}

\newtheorem{definition}[theorem]{Definition}

\newtheorem{proposition}[theorem]{Proposition}

\newcommand{\ket}[1]{|#1\rangle}

\def\>{\rangle}
\def\<{\langle}

\newcommand{\proj}[1]{|#1\rangle\!\langle#1|}

\begin{document}

\title{Uniform Additivity in Classical and Quantum Information}

 \author{Andrew Cross}
 \affiliation{IBM TJ Watson Research Center, Yorktown Heights, NY 10598, USA}
 \author{Ke Li}
 \affiliation{IBM TJ Watson Research Center, Yorktown Heights, NY 10598, USA}
\affiliation{Center for Theoretical Physics, Massachusetts Institute of Technology, Cambridge, MA 02139, USA}
 \author{Graeme Smith}
 \affiliation{IBM TJ Watson Research Center, Yorktown Heights, NY 10598, USA}

\date{\today}

\begin{abstract}
Information theory establishes the fundamental limits on data transmission, storage, and processing \cite{CoverThomas}.   Quantum information theory unites information theoretic ideas with an accurate quantum-mechanical description of reality to give a more accurate and complete theory with new and more powerful possibilities for information processing.  The goal of both classical and quantum information theory is to quantify the optimal rates of interconversion of different resources.  These rates are usually characterized in terms of entropies.   However, nonadditivity of many entropic formulas often makes finding answers to information theoretic questions intractable  \cite{DSS98,SS07, SY08, Hastings09,LWZG09,SS09a,CCH11,cubitt15}.  In a few auspicious cases, such as the classical capacity 
of a classical channel, the capacity region of a multiple access channel and the entanglement
assisted capacity of a quantum channel, additivity allows a full characterization of optimal rates.  Here we present a new mathematical property of entropic formulas, uniform additivity, that is both easily evaluated and rich enough to capture all known quantum additive formulas.  We give a complete characterization of uniformly additive functions using the linear programming approach to entropy inequalities.  In addition to all known quantum formulas, we find a new and intriguing additive quantity: the completely coherent information.  We also uncover a remarkable coincidence---the classical and quantum uniformly additive functions are identical;  the tractable answers in classical and quantum information theory are formally equivalent.  Our techniques pave the way for a deeper understanding of the tractability of information theory, from classical multi-user problems like broadcast channels to the evaluation of quantum channel capacities.
\end{abstract}

\maketitle

Entropies tell us how much information is stored in a system. As a result, the answers to information theoretic questions are usually found in terms of entropies evaluated on systems arising in optimal protocols.  For example, the communication capacity of a classical channel ${\cal N}$ that maps random variable $X$ to $Y$ is given by the maximization $C({\cal N}) = \max_X I(X;Y)$, where the mutual information $I(X;Y) = H(X) + H(Y) - H(XY)$ is a linear combination of entropies \footnote{The entropy 
of a random variable $X$ is $H(X) = -\sum_x p_x \log p_x$}. Similarly, the cost of transmitting a quantum state $\rho_A$ on system $A$ is its von Neumann entropy $H(A) = -\Tr \rho_A \log \rho_A$.  A noisy quantum communication channel ${\cal N}: A \rightarrow B$ 
can be mathematically extended to a unitary interaction $U: A \rightarrow BE$ of the input with an independent and inaccessible environment.  Such a channel can be applied to a state $\phi_{VA}$ to create a state $\rho_{VBE}$.  More generally, $V$ may have many subsystems, and we may use $\phi_{V_1...V_n A}$ to create $\rho_{V_1...V_nBE}$.  We can use such a state to generate an \emph{entropic formula}:
$f_{\alpha}(U_{\cN}) = \max_{\phi_{V_1...V_n A}} f_{\alpha}(U_{\cN},\phi_{V_1...V_n A})$
with
$f_{\alpha}(U_{\cN},\phi_{V_1...V_n A}) = \sum_{s \in \cP(V_1...V_n BE)} \alpha_s H(\rho_{s}),$
where $\cP(V_1...V_n BE)$ ranges over all collections of subsystems from $V_1...V_nBE$, and $H(\rho_{s})$ is the entropy of collection $s$.  We call the $V_1...V_n$  systems auxiliary variables.  Most operationally relevant quantities in quantum information can be expressed as a regularization of such a formula:
\begin{align}\label{Eq:RegularizedFormula}
f^{\infty}_{\alpha}(\cN) = \lim_{n \rightarrow\infty}\frac{1}{n}f_\alpha\left( \cN^{\otimes n}\right),
\end{align}
where $\cN^{\otimes n}$ is the $n$-fold parallel use of channel $\cN$.  The auxiliary variables in an entropic formula are usually related operationally to the structure of optimal protocols; for example, the optimal distribution $X$ that maximizes $C({\cal N}) = \max_X I(X;Y)$ 
to give the classical capacity defines a distribution of capacity-achieving error correcting codes.

\begin{figure}[htbp]
\includegraphics[width=3in]{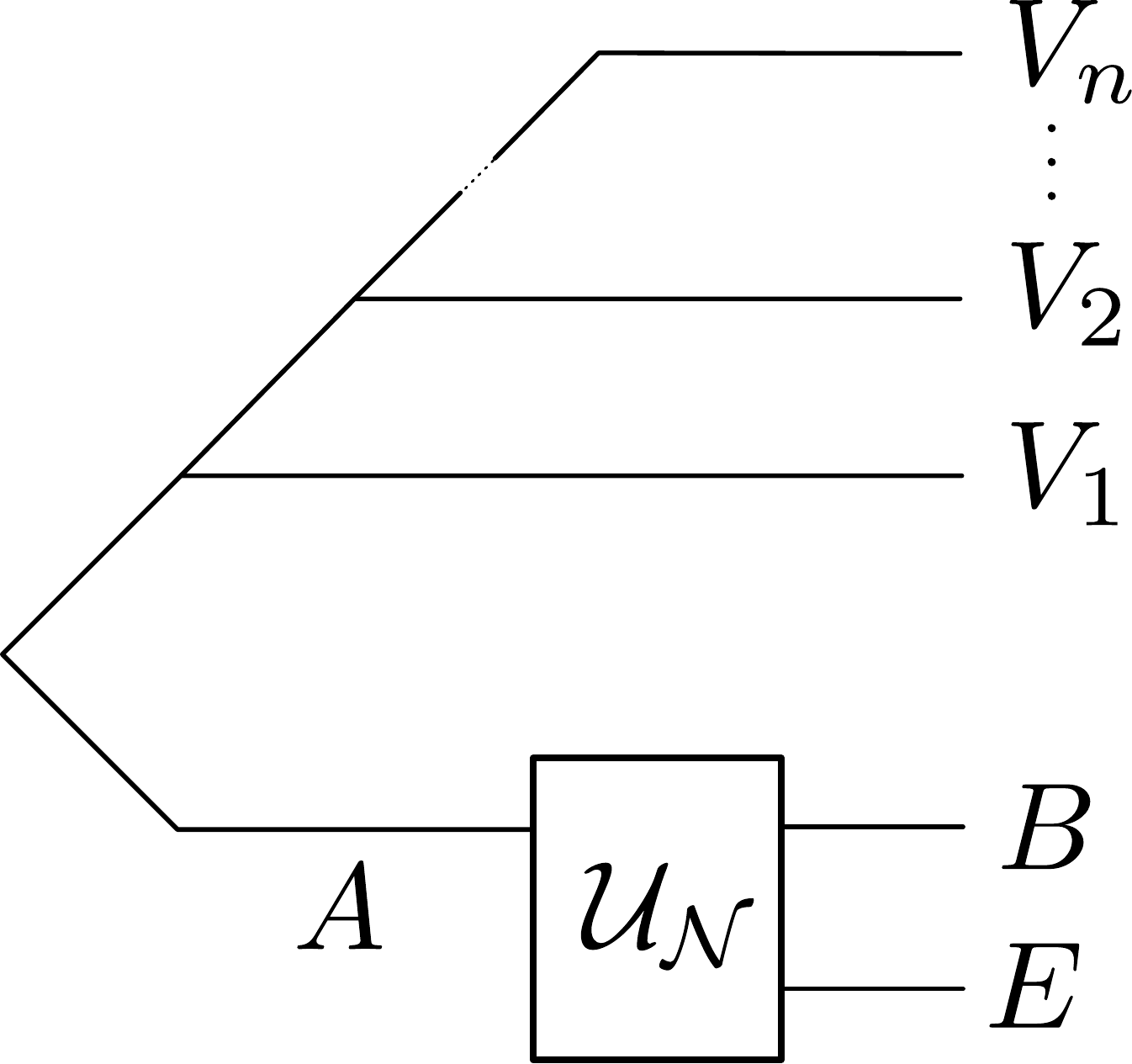}

\caption{Using a quantum channel to generate a quantum state.  A noisy quantum channel from input $A$ to output $B$ can always be thought of as a unitary interaction of the input with some inaccessible environment $E$.  We can generate a quantum
state from this interaction by creating $\phi_{V_1...V_n A}$ and acting on $A$ with $U_{\cal N}$, leading to the state $\rho_{V_1...V_n BE} = I \otimes U_{\cN}\phi_{V_1...V_n A} I \otimes U_{\cN}^\dagger$.}
\label{Fig:QuantumChannel}
\end{figure}

The infinite-dimensional optimization of Eq.(\ref{Eq:RegularizedFormula}), which is called a multi-letter formula,  is usually intractable.  In some rare cases additivity allows a substantial simplification.
An entropic formula $f_\alpha(U_\cN)$ is \emph{additive} if  $f_\alpha(U_\cN \otimes U_\cM) = f_\alpha(U_\cN) +  f_\alpha (U_\cM)$ for all channels $\cN$ and $\cM$.  
When $f_\alpha$ is additive, we have $f^\infty_{\alpha}(U_\cN) = f_{\alpha}(U_\cN)$, which is called a single-letter formula.  There are single-letter formulas for the classical capacity of a classical channel \cite{Shannon48}, the entanglement-assisted capacity of a quantum channel \cite{BSST02}, and the quantum capacity of a quantum channel with access to a special zero-capacity assistance channel\cite{SSW06}.  A single-letter formula often leads to a tractable means of evaluating a quantity.   

Many relevant entropic formulas are nonadditive, especially in the quantum setting\cite{DSS98,SS07, Hastings09,SS09a,CCH11}.  Optimal performance is thus captured only by a \emph{multi-letter} formula, which is 
intractable to evaluate.  As a result, many fundamental questions in quantum information theory remain open---the classical and quantum capacities of most channels are unknown, and even deciding if a quantum channel has nonzero quantum capacity seems insurmountable \cite{cubitt15}.

Entropy inequalities express relationships between entropies of different collections of subsystems that are satisfied for all states.  Subadditivity of entropy, for example, tells us that $H(A) + H( B) - H(AB) \geq 0$, or equivalently $I(A;B) \geq 0$.  
Its generalization, strong subadditivity\cite{LiebRuskai}, tells us that conditional mutual information is also positive: $I(A;B|C) = H(AC) + H(BC) - H(ABC)-H(C) \geq 0$.  The set of $(2^n -1)$-dimensional entropy vectors 
$\mathbf{v} = (H(X_1), ... H(X_n), ..., H(X_1...X_n))$ that can be realized by classical probability distributions on $X_1...X_n$ form a cone,  whose study in terms of linear programming was formalized in \cite{Yeung97}.  The larger cone of realizable quantum entropies was studied in \cite{Pip03}.  Entropy inequalitites are the key to proving additivity when it exists.

If $f_\alpha$ is an additive formula with one auxiliary variable \footnote{We focus on 1 auxiliary variable for simplicity. Multiple variables can be handled similarly.}, for any pair of channels $\cN,\cM$ and any state $\phi_{VA_1A_2}$, there must 
be a pair of states $\tilde{\phi}_{\tilde{V}A_1}$ and $\hat{\phi}_{\hat{V}A_2}$ such that 
\begin{align}\label{Eq:Subadditive}
f_\alpha\left(U_{\cN}\otimes U_{\cM}, \phi_{VA_1A_2}\right) \leq f_\alpha(U_{\cN}, \tilde{\phi}_{\tilde{V}A_1}) + f_\alpha(U_{\cM}, \hat{\phi}_{\hat{V}A_2}).
\end{align}
We call such a mapping $\phi_{VA_1A_2} \rightarrow (\tilde{\phi}_{\tilde{V}A_1}, \hat{\phi}_{\hat{V}A_2})$ a decoupling.   In principle, the appropriate decoupling may depend in an arbitrary way on the channels $\cN, \cM$ and the state $\phi_{VA_1A_2}$.  
In practice, useful decouplings are invariably what we call \emph{standard} decouplings, which have a very simple form and are described in Fig.~\ref{Fig:Decoupling}. 
Once we have fixed a decoupling and $f_\alpha$, we can use entropy inequalities to 
determine if Eq. (\ref{Eq:Subadditive})  is satisfied.   When $f_\alpha$ does satisfy Eq.~(\ref{Eq:Subadditive})  with $(\tilde{\phi}, \hat{\phi})$ defined by a standard decoupling $D$, we say $f_\alpha$ is uniformly subadditive with respect to $D$.  Since we also have 
$f_\alpha(U_\cN\otimes U_{\cM}, \tilde{\phi}\otimes \hat{\phi}) = f_\alpha(U_\cN, \tilde{\phi}) + f_\alpha(U_\cM, \hat{\phi}),$
  subadditivity implies that
\begin{align}
f_{\alpha}(U_\cN\otimes U_{\cM}) = f_{\alpha}(U_\cN) + f_{\alpha}(U_\cM)
\end{align}
 and we call $f_\alpha$ uniformly additive with respect to $D$.   All known proofs of quantum additivity proceed by choosing a standard decoupling and proving Eq. (\ref{Eq:Subadditive}) via entropy inequalities \cite{BSST02, DJKR06, SSW06}.

We have found all entropic formulas $f_\alpha$ that are uniformly additive with respect to standard decouplings.  We do this by enumerating all standard decouplings, and using the linear programming formulation of 
entropy inequalities to determine which $f_\alpha$ are uniformly subadditive for each decoupling.  Our approach captures all previously known examples of additive formulas and more.  This method opens a line of attack on a variety of questions,  from classical multiuser information theory to finding new classes of
channels with additive capacities, and clarifies when and where to expect quantum synergies like superactivation \cite{SY08}.

\begin{figure}[htbp]
\includegraphics[width=4in]{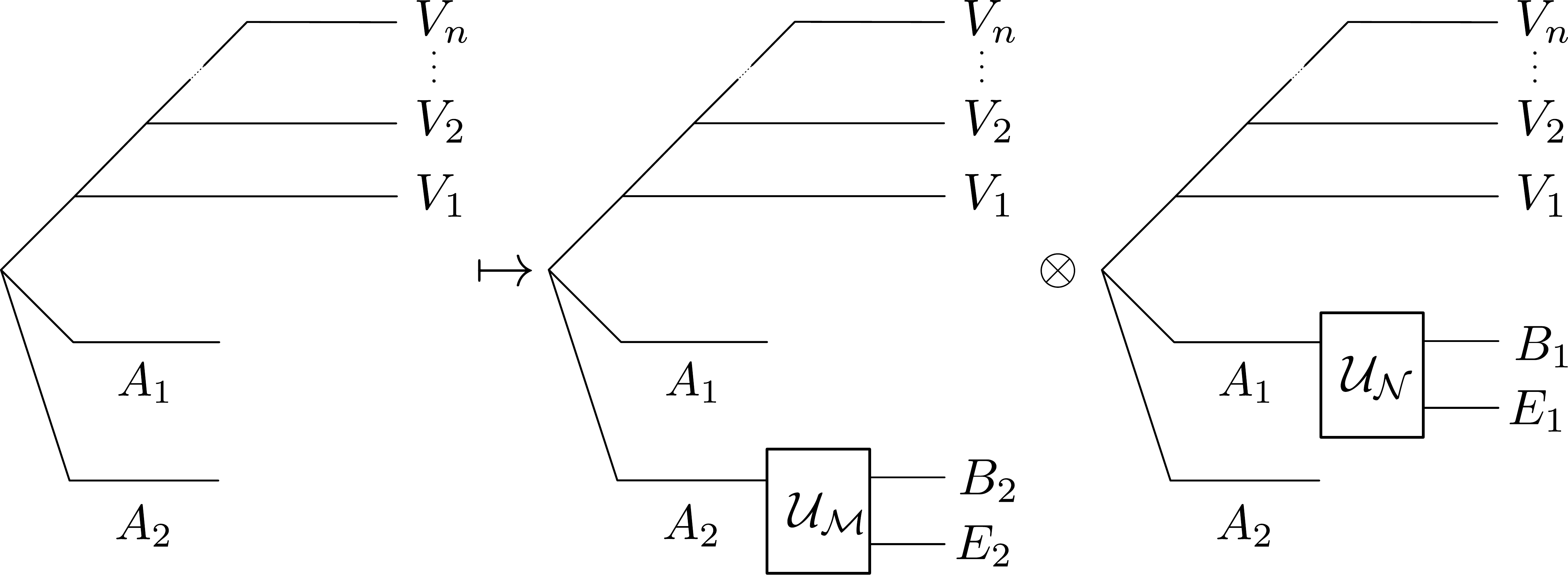}
\caption{Decoupling is the process of mapping one state that can be acted on by two channels into two separate states, each of which can be acted on by a single channel use.  It maps a state $\phi_{V_1...V_nA_1A_2}$ to two states, $\tilde{\phi}_{\tilde{V}_1...\tilde{V}_nA_1} $ and $\hat{\phi}_{\hat{V}_1...\hat{V}_nA_2} $.  Here $A_1$ and $A_2$ are the input spaces to $\cN$ and $\cM$, so that $U_{\cN}\otimes U_{\cM}$ can be applied to $\phi_{V_1...V_nA_1A_2}$ to make $\rho_{V_1...V_n B_1E_1B_2E_2}$, while  $U_{\cN}$ acts on $\tilde{\phi}_{\tilde{V}_1...\tilde{V}_n A_1}$ to make $\tilde{\rho}_{\tilde{V}_1...\tilde{V}_n B_1 E_1}$ and $U_{\cM}$ acts on $\hat{\phi}_{\hat{V}_1...\hat{V}_n A_2}$ to make $\hat{\rho}_{\hat{V}_1...\hat{V}_n B_2 E_2}$. For a \emph{standard} decoupling, the states  $\tilde{\phi}_{\tilde{V}_1...\tilde{V}_nA_1}$ and $\hat{\phi}_{\hat{V}_1...\hat{V}_nA_2}$ are constructed from $\phi_{V_1...V_nA_1A_2}$ as follows.  To obtain $\tilde{\phi}_{\tilde{V}_1...\tilde{V}_nA_1}$, we first apply $U_{\cM}$ to make
$\phi_{V_1...V_n A_1 B_2 E_2}$.  Given $\phi_{V_1...V_n A_1 B_2 E_2}$, we define $\tilde{V}_i$ to contain $V_i$. $B_2$ and $E_2$ are each either assigned to one of the $\tilde{V}_i$ (or perhaps traced out) to generate 
$\tilde{\phi}_{\tilde{V}_1...\tilde{V}_n A_1}$.    We define $\hat{\phi}_{\hat{V}_1...\hat{V}_n A_2}$ similarly.
  }
\label{Fig:Decoupling}
\end{figure}

\begin{figure}[htbp]
\includegraphics[width=3in]{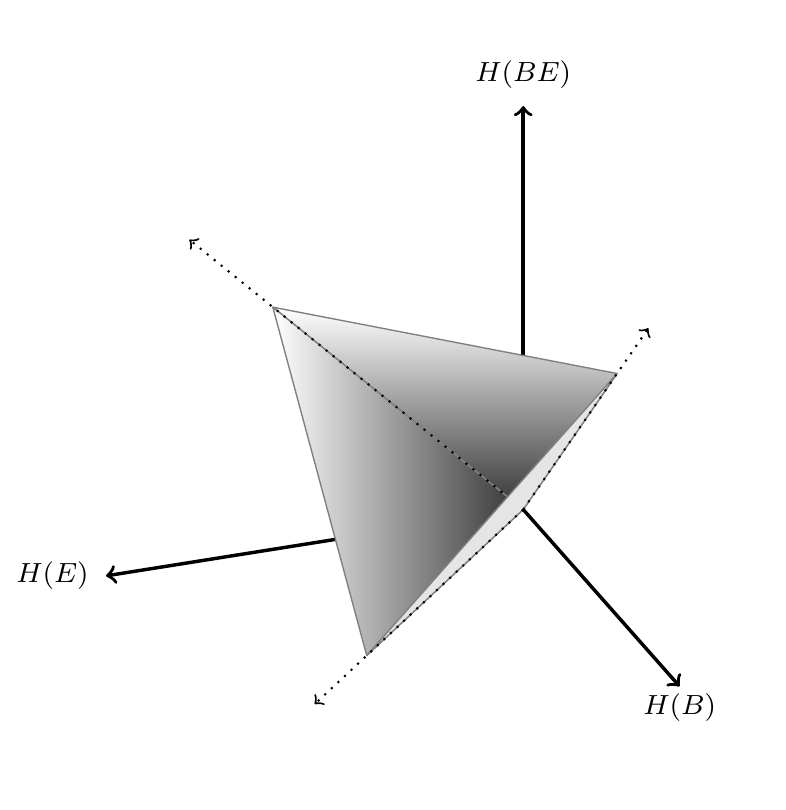}
\caption{Quantum Entropy Cone for two systems.  The entropies of a bipartite quantum state $\rho_{BE}$ form a vector $(H(B),H(E),H(BE))$.  The vectors of entropies that can be realized by a quantum state lie in a cone.  For two systems, the faces of this
cone are implied by strong subadditivity.  This is also true for $n=3$ systems, but for $n\geq 4$ we do not know whether the quantum entropy cone lies strictly inside the cone implied by strong subadditivity.   }
\label{Fig:EntropyCone}
\end{figure}

\begin{figure}[htbp]
\includegraphics[width=3in]{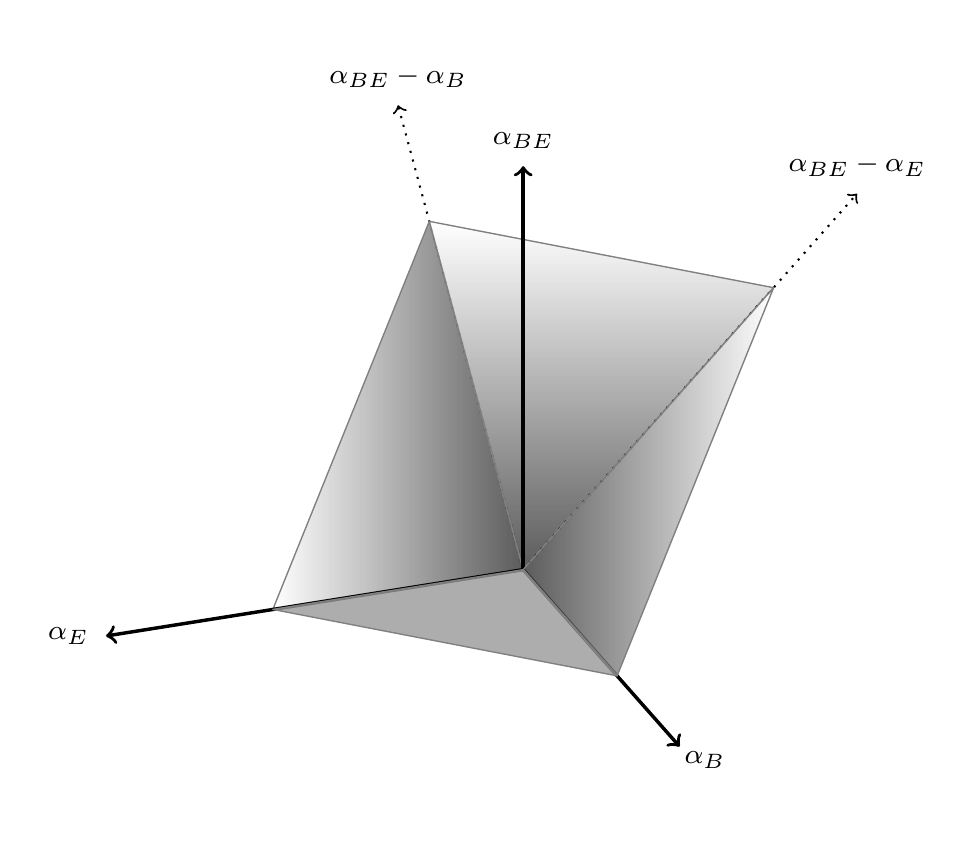}
\caption{Additivity cone.  
 Fixing a decoupling gives an entropy inequality that implies additivity.  We check whether this inequality is satisfied by using known additivity inequalites, as expressed
by the quantum entropy cone described in Figure \ref{Fig:EntropyCone}.  We find a cone of coefficients defining the entropy formulas that are uniformly additive with respect to the fixed decoupling.  The cone above is the additive cone for zero-auxiliary variable formulas.}
\label{Fig:AdditivityCone}
\end{figure}

Formulas with no auxiliary variables are particularly simple: $f_{\alpha}(U_\cN, \phi_A) = \alpha_BH(B) + \alpha_EH(E) + \alpha_{BE}H(BE)$.  Here we have only one standard decoupling to consider: $\phi_{A_1A_2} \rightarrow (\phi_{A_1},\phi_{A_2})$.  The conditions for uniform additivity in this case are

\begin{align}\label{Eq:ZeroVar}
\alpha_B + \alpha_{BE} &\geq 0\\
\alpha_E + \alpha_{BE} &\geq 0\nonumber\\
\alpha_B + \alpha_E + \alpha_{BE} &\geq 0\nonumber\\
\alpha_{BE} &\geq 0.\nonumber
\end{align}
These inequalities define a cone of $\alpha$s, which we refer to as a uniform additivity cone.  Eq.~(\ref{Eq:ZeroVar}) describes this cone in terms of its facets, but a cone can equally well be described in terms of extremal rays:  letting 
\begin{align}
\alpha_0 & = (1,0,0) \equiv H(B) \label{Eq:ExtRays}\\
\alpha_1 & = (0,1,0)\nonumber \equiv H(E)\\
\alpha_2 & = (0,-1,1)\nonumber\equiv H(B|E)\\
\alpha_3 & = (-1,0,1)\nonumber \equiv H(E|B),
\end{align}
$\alpha$ satisfies Eq.~(\ref{Eq:ZeroVar}) exactly when $\alpha = \sum_i \lambda_i \alpha_i$ with $\lambda_i \geq 0$.

Formulas with one auxiliary variable require us to consider multiple decouplings, capturing the choice of $\tilde{V}$ and $\hat{V}$ in the decoupling map $D: \phi_{VA_1A_2} \rightarrow (\tilde{\phi}_{\tilde{V}A_1}, \hat{\phi}_{\hat{V}A_2})$. 
A standard decoupling always has $\tilde{V} = \tilde{M}_2V$ with $\tilde{M}_2$ chosen from $\{\varnothing, B_2,E_2,B_2E_2\}$ and $\hat{V} = \hat{M}_1V$ with $\hat{M}_1$ chosen from $\{\varnothing, B_1,E_1,B_1E_1\}$.  We can parametrize these
by $(a,b)$, with $a$ and $b$ running from $0$ to $3$.  We take advantage of two simplifications that can be made without loss of generality.   First, given $f_\alpha$, $\alpha = (\alpha^\varnothing, \alpha^V)$ with $\alpha^\varnothing = (\alpha_B,\alpha_E,\alpha_{BE})$ and 
$\alpha^V = (\alpha_V, \alpha_{BV},\alpha_{EV},\alpha_{BEV})$, we can define $f^\varnothing_{\alpha^\varnothing}$ and $f^V_{\alpha^V}$ such that $f_\alpha$ is uniformly additive with respect to decoupling $(a,b)$ if and only if
$f^\varnothing_{\alpha^\varnothing}$ is uniformly additive with respect to the decoupling $\phi_{A_1A_2}\rightarrow (\phi_{A_1},\phi_{A_2})$ and $f^V_{\alpha^V}$ is uniformly additive 
with respect to $(a,b)$.  Second, these formulas have two useful symmetries that reduce 
the number of decouplings we must consider: 1) for any additive formula, we get a similar additive formula by exchanging $B$ and $E$ and 2)  $f^V_{\alpha^V}$  with $\alpha^V = (\alpha_V, \alpha_{BV}, \alpha_{EV}, \alpha_{BEV})$ is equivalent via purification 
of the quantum state to $f^V_{\tilde{\alpha}^V}$ with $\tilde{\alpha}^V = (\alpha_{BEV}, \alpha_{EV},\alpha_{BV},\alpha_V)$.  This leaves only $5$ inequivalent decouplings to be considered.

Figure \ref{Fig:OneVarTable} describes the functions $f^V_{\alpha^V}$ that are uniformly additive with respect to the $5$ inequivalent decouplings.  They are positive linear combinations\footnote{i.e., linear combinations with positive coefficients} of
the extreme rays in the corresponding row of the table.  The uniformly additive functions with respect to decoupling $(a,b)$ are the sum of any $f^\varnothing_{\alpha^\varnothing}$ satisfying Eq. (\ref{Eq:ZeroVar}) and such an $f^V_{\alpha^V}$ found from Figure \ref{Fig:OneVarTable}.

\begin{figure}[htbp]

\begin{tabular}{|l|c|c|c|c|c| l |}
\hline
case & (a,b) & $\hat{M}_1$ & $\tilde{M}_2$ & equivalents & Additive Cone & Extreme Rays\\
\hline
1. & (3,3) & $B_1E_1$ & $B_2E_2$ & (0,0) & $\begin{array}{c}
\alpha_V + \alpha_{BV}+\alpha_{EV}  \geq 0\\
\alpha_V + \alpha_{BV}  \geq 0\\
\alpha_V + \alpha_{EV} \geq 0\\
\alpha_V \geq 0\end{array}$  & 
$\begin{array}{c}
 -H(E|BV) \\
-H(E|V) \\
 -H(B|EV) \\
-H(B|V)
\end{array}$\\
\hline
2. & (3,2) & $B_1E_1$ & $E_2$ &  $\begin{array}{c}(2,3), (3,1)\\ (1,3), (1,0),(0,1)\\ (2,0),(0,2) \end{array}$ &  $\begin{array}{c} \alpha_{BV} \leq 0\\ \alpha_V + \alpha_{BV}\geq 0\end{array}$  & $\begin{array}{c}  -H(BE|V) \\ \pm H(B|EV) \\ -H(B|V)\end{array}$\\
\hline 
3. & (3,0) & $B_1E_1$ & $\varnothing$ & (0,3) & $\begin{array}{c} \alpha_{EV} \leq 0\\ \alpha_{BV} \leq 0\end{array}$ & $\begin{array}{c}  H(E|BV) \\  -H(E|V) \\  \pm H(BE|V) \end{array}$\\
\hline
4. & (1,1) &  $B_1$ & $B_2$ & (2,2)  & $\begin{array}{c} \alpha_{EV} = 0\\ \alpha_V \geq 0\\ \alpha_{BEV} \geq 0\end{array}$  & $\begin{array}{c} -H(B|V) \\  H(E|BV)\end{array}$\\
\hline
5. & (1,2) & $B_1$ & $E_2$ & (2,1)& $\begin{array}{c} \alpha_{BEV} \geq 0\\ \alpha_V \geq 0\end{array}$ & $\begin{array}{c}\pm [H(EV)-H(BV)] \\ H(E|BV) \\  -H(E|V) \end{array}$\\
\hline
\end{tabular}

\caption{Functions $f^V_{\alpha^V}$ that are uniformly subadditive with respect to the $5$ inequivalent standard decouplings.  Fixing a decoupling $D$, a single-auxiliary variable $f_{\alpha}$ is uniformly subadditive with respect to $D$ exactly when it can 
be written as a sum of $f^\varnothing_{\alpha^\varnothing}$ satisfying Eq.\ref{Eq:ZeroVar}  and $f^V_{\alpha^V}$ that is a positive linear combination of the extreme rays in the row corresponding to $D$. Multiple auxiliary variables are all found similarly.\label{Fig:OneVarTable}}
\end{figure}

We find many familar additive quantities in this way.  For example, maximum output entropy ($\max_{\phi_A}H(B)$) satisfies Eq. (\ref{Eq:ZeroVar}).  The quantity $-H(B|V)$ was shown to be additive in \cite{DJKR06}, and later refered to as reverse
 coherent information\cite{GPLS09}.  Since $H(B)$ satisfies Eq. (\ref{Eq:ZeroVar}) and $-H(B|V)$ is uniformly additive with respect to multiple decouplings, so is 
their sum $H(B)-H(B|V) = I(B;V)$, whose maximization gives the entanglement assisted capacity.

One extreme ray of the $(1,2)$  decoupling's additive cone is particularly intriguing: $I^{cc}(\cN) = \max_{\phi_{VA}}[H(VB) - H(VE)]$.  We call this quantity the completely coherent
information, since its relationship to the coherent information $I^{\rm coh}(\cN) = \max_{A} [H(B)-H(E)]$ is similar to the relationship between completely positive and positive maps.  The version of this quantity evaluated on states was 
shown in \cite{OW05} to be a lower bound on the communication cost of exchanging the $B$ and $E$ systems, but it was not realized that it is additive.  We also show that $I^{cc}$ is also an \emph{upper}
bound for the jointly achievable quantum communication rate from $A$ to either $B$ or $E$.  We have not found a clear operational interpretation of this quantity. 
 
%*******leave out for now??: Generalized: $H(BA) - H(EA)$, but for all functions additive on tensor products.

We now consider formulas with multiple auxiliary variables.  For concreteness, suppose we have some formula depending on two auxiliary variables $V_1$ and $V_2$.  A standard decoupling is a mapping from a state $\phi_{V_1V_2A_1A_2}$ to two states
$\tilde{\phi}_{\tilde{V}_1\tilde{V}_2 A_1}$ and $\hat{\phi}_{\hat{V}_1\hat{V}_2 A_2}$ that we get by choosing to incorporate (or not) $B_2$ and $E_2$ into one of $\tilde{V}_1$ and  $\tilde{V}_2$ (and similarly for $B_1$, $E_1$ in $\hat{V}_1$ and  $\hat{V}_2$).  Since $\tilde{V}_1$ and  $\tilde{V}_2$ should be non-overlapping, it is necessary to impose some consistency on the decouplings $(a_1,b_1)$ and $(a_2,b_2)$.  These also give rise to a third decoupling, which we call $(a^\star,b^\star)$, that tells us which systems get 
included in the joint systems $\tilde{V}_1\tilde{V}_2$ and $\hat{V}_1\hat{V}_2$.

 In this case it  is possible to separate the variables much as we did in the single-variable case.  Indeed, any $f_\alpha$ with $\alpha = (\alpha^\varnothing, \alpha^{V_1}, \alpha^{V_2}, \alpha^{V_1V_2})$ \footnote{Here $\alpha^\varnothing = (\alpha_B,\alpha_E\alpha_{BE})$, $\alpha^{V_1} = (\alpha_{V_1}, \alpha_{BV_1},\alpha_{EV_1}\alpha_{BEV_1})$, $\alpha^{V_2} = (\alpha_{V_2}, \alpha_{BV_2},\alpha_{EV_2}\alpha_{BEV_2})$ and $\alpha^{V_1V_2} = (\alpha_{V_1V_2}, \alpha_{BV_1V_2},\alpha_{EV_1V_2}\alpha_{BEV_1V_2})$.} is uniformly additive with respect to decoupling $\{(a_1,b_1),(a_2,b_2)\}$ exactly when $f^\varnothing_{\alpha^\varnothing}$, $f^{V_1}_{\alpha^{V_1}}$, $f^{V_2}_{\alpha^{V_2}}$, and $f^{V_1V_2}_{\alpha^{V_1V_2}}$ are uniformly additive with respect to their respective decouplings.  The same is true for 
more auxiliary variables.  For any number of auxiliary variables, all $f_\alpha$ uniformly additive with respect to standard decouplings can be constructed from Figure \ref{Fig:OneVarTable} and Eq.~(\ref{Eq:ZeroVar}).

%\section{Discussion}

%\textbf{[Para 2]} Classical vs quantum same ''reason'' for so many coincidences.\\

Surprisingly, carrying out the same analysis as above for classical states and channels yields \emph{exactly} the same set of uniformly additive functions.  This is 
in spite of the fact that the classical and quantum entropy cones do not coincide.  This coincidence of uniformly additive functions may explain a well-known phenomenon: 
Formulas that solve classical information theory problems often tend to have corresponding quantum formulas that solve an appropriately coherified version of the problem
\footnote{Examples of this include 1) the correspondence between classical capacity of a classical channel and the entanglement assisted capacity of a quantum channel, 2) the connection 
between Slepian-Wolf and state merging, and 3) the correspondence between Csiszar-Korner solution to the broadcast channel with confidential messages and the recent
analysis of the quantum one-time pad.}.  In these cases, the classical and quantum problems have a solution for the same reason: the existence of an appropriately additive
formula whose additivity proofs are formally equivalent.  It would be very nice to formalize this apparent correspondence.

We are currently exploring the application of our techniques to finding special classes of channels that have additive capacities.  We have identified a
new criterion for the additivity of coherent information: informational degradability.  We say a channel is informationally degradable if for any input state $\phi_{VA}$ we have $I(V;B) \geq I(V;E)$.   This
class includes degradable channels.  We suspect informational degradability is the \emph{only } single-letter entropic constraint on a channel that implies this additivity.  
We have also found a set of entropic constraints that imply a state is of the c-q form, which should be useful for studying classical and private capacities of quantum channels.

We have identified the limits of the techniques used in all known instances of quantum additivity.  There \emph{are} some classical formulas that are additive but not uniformly additive (e.g., minimum output entropy of a classical channel).
 Proving additivity in these cases requires knowledge of the optimizing state (in the case of minimum output entropy of a quantum channel, the optimal state
is a pure state, which for classical channels is also a product state.).  One potential path to new quantum additive formulas beyond what we have found is to better understand the optimizing state in an entropic formula.  At this 
point we know of no examples where this can be done, but they may well exist.

\section{Methods}
We now argue that Eq~(\ref{Eq:ZeroVar}) captures all uniformly additive formulas with no auxiliary variables.  To begin, for a zero auxilliary variable $f_\alpha$, we define
 \begin{align}
\Delta^\varnothing(\alpha,U_\cN\otimes U_\cM,\phi_{A_1A_2}) & = f_{\alpha}(U_{\cN}, \phi_{A_1}) + f_{\alpha}(U_{\cM}, \phi_{A_2}) - f_{\alpha}(U_{\cN}\otimes U_{\cM}, \phi_{A_1A_2})\\
& = \alpha_B I(B_1;B_2) + \alpha_E I(E_1;E_2) + \alpha_{BE} I(B_1E_1;B_2E_2),
\end{align}
so that $f_\alpha$ is uniformly additive with respect to the standard decoupling exactly when $\forall \cN,\cM,\phi_{A_1A_2}$ we have $\Delta^\varnothing(\alpha,U_\cN\otimes U_\cM,\phi_{A_1A_2}) \geq 0$.
We make use of the alternate characterization of Eq.(\ref{Eq:ZeroVar}) in terms of extremal rays, Eq.~(\ref{Eq:ExtRays}).  It is easy to verify that the $\alpha$s associated with each of the extremal rays $H(B)$, $H(E)$, $H(E|B)$, and $H(B|E)$ lead to positive $\Delta^\varnothing$s.  For
example, $H(B)$ corresponds to $(\alpha_B, \alpha_E, \alpha_{BE}) = (1,0,0)$ and $\Delta^\varnothing(\alpha,U_\cN\otimes U_\cM,\rho_{A_1A_2}) = I(B_1;B_2) \geq 0$, while $H(B|E)$ corresponds to $(\alpha_B, \alpha_E, \alpha_{BE}) = (0,-1,1)$ and gives
$\Delta^\varnothing(\alpha,U_\cN\otimes U_\cM,\rho_{A_1A_2}) = I(B_1E_1;B_2E_2) - I(E_1;E_2)$, which is also positive for all $\rho_{A_1A_2}$.  $H(E)$ and $H(E|B)$ follow mutatis mutandis.  Eq. (\ref{Eq:ZeroVar}) is thus a sufficient condition for uniform additivity.  To see that it is also a 
necessary condition, we find states (in fact, classical distributions) $p^0$, $p^1$, $p^2$, $p^3$ and channels $\cN$, $\cM$ such that   
\begin{align}
\Delta^\varnothing(\alpha,U_\cN\otimes U_\cM,p^0) &= \alpha_B + \alpha_{BE}\nonumber\\
\Delta^\varnothing(\alpha,U_\cN\otimes U_\cM,p^1) &= \alpha_E + \alpha_{BE}\nonumber\\
\Delta^\varnothing(\alpha,U_\cN\otimes U_\cM,p^2) &= \alpha_B+ \alpha_E + \alpha_{BE}\nonumber\\
\Delta^\varnothing(\alpha,U_\cN\otimes U_\cM,p^3) &=  \alpha_{BE}\nonumber.
\end{align}
This shows that for any $\alpha$ that doesn't satisfy Eq.~(\ref{Eq:ZeroVar}) there are states and channels such that $\Delta^\varnothing(\alpha,U_\cN\otimes U_\cM,p) < 0$.  Thus, Eq.~(\ref{Eq:ZeroVar}) are both necessary and sufficient for uniform additivity.

 Uniform additivity with one auxiliary variable requires us to consider 5 inequivalent decouplings.  Fixing a decoupling $(a,b)$ that maps $\phi_{VA_1A_2} \rightarrow (\tilde{\phi}_{\tilde{V}A_1}, \hat{\phi}_{\hat{V}A_2})$
define
\begin{align}
\Delta^{(a,b)}(\alpha, U_\cN \otimes U_\cM, \phi_{VA_1A_2}) & = f_\alpha(U_\cN, \tilde{\phi}_{\tilde{V}A_1}) + f_\alpha(U_\cM, \hat{\phi}_{\hat{V}A_2}) - f_\alpha(U_\cN\otimes U_\cM, {\phi}_{VA_1A_2})
\end{align}
so that $f_\alpha$ is uniformly additive with respect to $(a,b)$ exactly when for all $U_\cN$, $U_\cM$, $\phi_{VA_1A_2}$ we have $\Delta_{(a,b)}(\alpha, U_\cN \otimes U_\cM, \phi_{VA_1A_2}) \geq 0$.  Finding the uniformly subadditive $f_\alpha$ is greatly simplified 
through the separation of variables: letting $\alpha  = (\alpha^\varnothing, \alpha^V)$ with $\alpha^\varnothing = (\alpha_B, \alpha_E, \alpha_{BE})$ and $\alpha^V = (\alpha_V, \alpha_{BV}, \alpha_{EV}, \alpha_{BEV})$ and defining
\begin{align}
\Delta^\varnothing(\alpha^\varnothing, U_\cN \otimes U_\cM, \phi_{A_1A_2}) & = f_{\alpha^\varnothing}(U_{\cN}, \phi_{A_1}) + f_{\alpha^\varnothing}(U_{\cM}, \phi_{A_2}) - f_{\alpha^\varnothing}(U_{\cN}\otimes U_{\cM}, \phi_{A_1A_2}) \\
\Delta^{V,(a,b)}(\alpha^V, U_\cN \otimes U_\cM, \phi_{VA_1A_2}) & = f_{(0,\alpha^V)}(U_\cN, \tilde{\phi}_{\tilde{V}A_1}) + f_{(0,\alpha^V)}(U_\cM, \hat{\phi}_{\hat{V}A_2}) - f_{(0,\alpha^V)}(U_\cN\otimes U_\cM, {\phi}_{VA_1A_2})
\end{align}
we have 
\begin{align}
\Delta^{(a,b)}(\alpha, U_\cN \otimes U_\cM, \phi_{VA_1A_2})  & = \Delta^\varnothing(\alpha^\varnothing, U_\cN \otimes U_\cM, \phi_{A_1A_2}) + \Delta^{V,(a,b)}(\alpha^V, U_\cN \otimes U_\cM, \phi_{VA_1A_2}).
\end{align}
Furthermore, $\Delta^{(a,b)}(\alpha, U_\cN \otimes U_\cM, \phi_{VA_1A_2}) \geq 0$ for all $U_\cN$,  $ U_\cM$, $\phi_{VA_1A_2}$ exactly when $\Delta^\varnothing(\alpha^\varnothing, U_\cN \otimes U_\cM, \phi_{A_1A_2}) \geq 0$ for all $U_\cN$,  $ U_\cM$, $\phi_{VA_1A_2}$ and $\Delta^{V,(a,b)}(\alpha^V, U_\cN \otimes U_\cM, \phi_{VA_1A_2})\geq 0$ for all  $U_\cN$,  $U_\cM$, and  $\phi_{VA_1A_2}$.  We have already characterized 
when $\Delta^\varnothing(\alpha,U_\cN \otimes U_\cM, \phi_{VA_1A_2}) \geq 0$ in the previous paragraph, and we can determine the $\alpha^V$ such that $\Delta^{V,(a,b)}(\alpha^V, U_\cN \otimes U_\cM, \phi_{VA_1A_2}) \geq 0$ for all $U_\cN$,  $U_\cM$, and  $\phi_{VA_1A_2}$ in a similar way (either by direct computation or linear programming).

\section{Acknowledgements}
KL acknowledges  NSF grants CCF-1110941 and CCF-1111382 and GS acknowledges NSF grant CCF-1110941.

%\clearpage

%\bibliographystyle{alpha}
\bibliographystyle{apsrev}
\bibliography{shortq}

\clearpage
\appendix

\section{Notation and Background}
\label{sec:Notation and Background}
For any collection of systems $X_1\dots X_n$, let $\cP(X_1\dots X_n)$ be the power set of this collection
\begin{align}
\cP(X_1\dots X_n) &= \left\{ X_1^{u_1}\dots X_n^{u_n}| (u_1,\dots,u_n) \in \{0,1\}^n  \right\}.
\end{align}
We study channels $U_{\cN}: A \rightarrow BE$ and are interested in formulas $f_\alpha(U_{\cN})$ that are maximizations of linear combinations of entropies involving auxiliary variables $V_1\dots V_n$
\begin{align}
f_\alpha(U_{\cN}) = \max_{\phi_{V_1\dots V_n A}} f_\alpha(U_{\cN}, \phi_{V_1\dots V_n A}),
\end{align}
where the linear entropic quantity $f_\alpha(U_{\cN},\phi_{V_1\dots V_n A})$ is given by
\begin{align}
f_\alpha(U_{\cN}, \phi_{V_1\dots V_n A})  = \sum_{s \in \cP(V_1\dots V_n BE)} \alpha_s H\left(s,\rho_{V_1\dots V_n BE}\right)
\end{align}
where $\rho_{V_1\dots V_n BE} =( I \otimes U_{\cN}) \phi_{V_1\dots V_n A} (I \otimes U_{\cN}^\dagger)$ is the channel output state and $H\left(s,\rho_{V_1\dots V_n BE}\right)$ is the entropy of the reduced state corresponding to systems $s$.

\section{General Considerations}
\label{sec:General Considerations}
We are interested in understanding when
\begin{align}
f_\alpha(U_\cN\otimes U_\cM) & = f_\alpha(U_\cN) + f_\alpha(U_\cM).
\end{align}
In order to do this, we study mappings from a state $\phi_{V_1...V_n A_1A_2}$ that can be acted on by $U_\cN\otimes U_\cM$ to two states: $\tilde{\phi}_{\tilde{V}_1\dots \tilde{V}_n A_1}$, which can be acted on by $\cN$ and
$\hat{\phi}_{\hat{V}_1\dots \hat{V}_n A_2}$ which can be acted on by $\cM$.  We call such a mapping,  $\cD: \phi_{V_1\dots V_n A_1 A_2} \rightarrow (\tilde{\phi}_{\tilde{V}_1\dots \tilde{V}_n A_1},\hat{\phi}_{\hat{V}_1\dots \hat{V}_n A_2} )$  a {\em decoupling}.

There are two important types of decouplings that we consider: standard decouplings, and consistent decouplings. Both types of decouplings construct $\tilde{\phi}_{\tilde{V}_1\dots \tilde{V}_n A_1}$ from relabling the systems of
$I\otimes U_\cM \phi_{V_1\dots V_n A_1 A_2} I\otimes U_\cM^\dagger$ and construct $\hat{\phi}_{\hat{V}_1\dots \hat{V}_n A_2}$ by relabling the systems of $I\otimes U_\cN \phi_{V_1\dots V_n A_1 A_2} I\otimes U_\cN^\dagger$.
For a \emph{standard decoupling}, we have
$\tilde{V}_i  = \tilde{M}_2 V_i$ and $\hat{V}_i  = \hat{M}_1 V_i$ with $\tilde{M}_2 \in \cP(B_2E_2)$ and $\hat{M}_1 \in \cP(B_1E_1)$. For a \emph{consistent decoupling},
we require less: $\tilde{V}_i \in \cP(V_1...V_n B_2E_2)$ with $\tilde{V}_i \cap \tilde{V}_j = \varnothing$ and  $\hat{V}_i \in \cP(V_1...V_n B_1E_1)$ with $\hat{V}_i \cap \hat{V}_j = \varnothing$.

We say that  $f_\alpha(U_{\cN}, \phi_{V_1\dots V_n A})$ is \emph{uniformly} subadditive with respect to decoupling $\cD$ if for all $\cN_1$, $\cN_2$, and $\phi_{V_1\dots V_n A_1 A_2}$ we have

\begin{align}
f_\alpha(U_{\cN_1} \otimes U_{\cN_2}, \phi_{V_1\dots V_n A_1 A_2}) \leq f_\alpha(U_{\cN_1}, \tilde{\phi}_{\tilde{V}_1\dots \tilde{V}_n A_1}) + f_\alpha(U_{\cN_2}, \hat{\phi}_{\hat{V}_1\dots \hat{V}_n A_2}).
\end{align}

The following quantity will be useful:
\begin{equation}\begin{split} \label{eq:Delta-def}
\Delta(f_\alpha, U_{\cN_1},& U_{\cN_2},\phi_{V_1\dots V_n A_1 A_2},\tilde{\phi}_{\tilde{V}_1\dots \tilde{V}_n A_1}, \hat{\phi}_{\hat{V}_1\dots \hat{V}_n A_2}) \\
 =  f_\alpha(U_{\cN_1},& \tilde{\phi}_{\tilde{V}_1\dots \tilde{V}_n A_1}) + f_\alpha(U_{\cN_2}, \hat{\phi}_{\hat{V}_1\dots \hat{V}_n A_2}) -f_\alpha(U_{\cN_1} \otimes U_{\cN_2}, \phi_{V_1\dots V_n A_1 A_2}).
\end{split}\end{equation}

Defined in this way, $\Delta$ is linear in $f_\alpha$, so if we have
\begin{align*}
\Delta\left(f_{\alpha_1}, U_{\cN_1},U_{\cN_2},\phi_{V_1\dots V_n A_1 A_2},\tilde{\phi}_{\tilde{V}_1\dots \tilde{V}_n A_1}, \hat{\phi}_{\hat{V}_1\dots \hat{V}_n A_2}\right) &\geq 0\\
\Delta\left(f_{\alpha_2}, U_{\cN_1},U_{\cN_2},\phi_{V_1\dots V_n A_1 A_2},\tilde{\phi}_{\tilde{V}_1\dots \tilde{V}_n A_1}, \hat{\phi}_{\hat{V}_1\dots \hat{V}_n A_2}\right) &\geq 0
\end{align*}
then for $\lambda_1,\lambda_2\geq 0$ we also have
\begin{align}\label{Eq:PosLamb}
\Delta\left(\lambda_1 f_{\alpha_1} + \lambda_2 f_{\alpha_2}, U_{\cN_1},U_{\cN_2},\phi_{V_1\dots V_n A_1 A_2},\tilde{\phi}_{\tilde{V}_1\dots \tilde{V}_n A_1}, \hat{\phi}_{\hat{V}_1\dots \hat{V}_n A_2}\right) &\geq 0.
\end{align}

For the standard or consistent decouplings, the $\Delta$ function defined in Eq.~(\ref{eq:Delta-def}) depends only on the
decoupling $\cD$, the entropy formula $f_\alpha$ and  the state
\begin{align}
\label{eq:state-gen}
\rho_{V_1\ldots V_nB_1E_1B_2E_2}=(I\otimes U_{\cN_1} \otimes U_{\cN_2})\phi_{V_1\ldots V_nA_1A_2}(I\otimes U_{\cN_1}^\dagger \otimes U_{\cN_2}^\dagger).
\end{align}
So we abbreviate it as
$\Delta^{\cD}(\alpha,\rho_{V_1\ldots V_nB_1E_1B_2E_2})$. It is easy to see that any state $\rho_{V_1\ldots V_nB_1E_1B_2E_2}$ can be written the form of Eq.~(\ref{eq:state-gen}), with appropriate $ U_{\cN_1}$, $ U_{\cN_2}$ and $\phi_{V_1\ldots V_nA_1A_2}$. Thus $f_\alpha(U_{\cN}, \phi_{V_1\dots V_n A})$ is uniformly subadditive with respect to the decoupling $\cD$ if and only if
\[
  \forall\,\rho_{V_1\ldots V_nB_1E_1B_2E_2},\quad \Delta^{\cD}(\alpha,\rho_{V_1\ldots V_nB_1E_1B_2E_2})\geq 0.
\]

\section{non-infinite functions that are uniformly subadditive}

We will restrict our attention to entropic formulas $f_\alpha$ that are not always infinite: there is at least one $U_\cN$ such that $f_\alpha(U_\cN) < \infty$.  This requirement leads to a particularly nice structure on the $\alpha$'s of
a uniformly additive function.

\begin{lemma}\label{Lemma:SigmaZero}
Let $f_\alpha(U_{\cal N},\phi_{V_1\dots V_nA})$ satisfy
\begin{align*}
f_{\alpha}(U_{\cal N}) = \max_{\phi_{V_1\dots V_nA}} f_\alpha(U_{\cal N},\phi_{V_1\dots V_nA}) < \infty
\end{align*}
for \emph{some} $U_{\cal N}$ and
\begin{align*}
\min_{\rho}\Delta^{\cD}(\alpha,\rho_{V_1\ldots V_nB_1E_1B_2E_2}) \geq 0,
\end{align*}
for a standard decoupling $\cD$. In words, $f_\alpha$ is bounded and uniformly subadditive with respect to the standard decoupling $\cD$. Then for all non-empty $t \in {\cal P}(V_1 \dots V_n)$,
\begin{align*}
\eta_t:=\sum_{s \in {\cal P}(BE)}\alpha_{s,t} = 0.
\end{align*}
\end{lemma}

\begin{proof}
For a channel $\cN$ such that $f_{\alpha}(U_{\cN})< \infty$, considering a state of the form $\rho_{V_1\ldots
V_n}^{\otimes k}\otimes \rho_A$, we have
\begin{equation*}\begin{split}
   &f_\alpha(U_\cN,\rho_{V_1\ldots V_n}^{\otimes k}\otimes \rho_A) \\
  =&f_\alpha(U_\cN,\proj{0}_{V_1}\otimes\ldots\otimes\proj{0}_{V_n}\otimes \rho_A)+k \sum_{t \in  {\cal P}(V_1 \dots V_n)} \eta_t H(t,\rho_{V_1...V_n}).
\end{split}\end{equation*}
So we must have
\begin{align}\label{Eq:SumLEQ}
  \sum_{t \in  {\cal P}(V_1 \dots V_n)} \eta_t H(t,\rho_{V_1...V_n}) \leq 0,
\end{align}
because otherwise, the quantity $f_\alpha(U_\cN,\rho_{V_1\ldots V_n}^{\otimes k}\otimes \rho_A)$ would go to $\infty$ as $k \rightarrow \infty$. Now, in order for $f_\alpha$ to be uniformly subadditive with respect to the standard decoupling $\cD$, we need
\begin{align*}
\Delta^{\cD}\left(\alpha, \rho_{V_1...V_nB_1E_1B_2E_2} \right) \geq 0
\end{align*}
for all $\rho_{V_1...V_nB_1E_1B_2E_2}$. This implies
\begin{align}\label{Eq:SumGEQ}
\Delta^{\cD}\left(\alpha, \rho_{V_1...V_n}\otimes\proj{0000}_{B_1E_1B_2E_2} \right) = \sum_{t \in {\cal P}(V_1\dots V_n)} \eta_t H(t,\rho_{V_1\dots V_n}) \geq 0,
\end{align}
where we have used the fact that $H(s_1\tilde{X})+H(s_2\hat{X})-H(s_1s_2X)=H(X)$ for this state and any subset $X$ of systems $V_1\dots V_n$. Eq.~(\ref{Eq:SumLEQ}) and Eq.~(\ref{Eq:SumGEQ}) together imply that
\begin{align}
 \sum_{t \in {\cal P}(V_1\dots V_n)} \eta_t H(t,\rho_{V_1\dots V_n}) = 0
\end{align}
for all $\rho_{V_1\dots V_n}$.
This implies that each $\eta_t = 0$, by uniqueness results from the classical literature (Theorem 1 of \cite{Yeung97}).
\end{proof}

We let
\begin{align}
{\cal F} = \left\{ \alpha\,|\, \exists\, U_{\cN},\ f_\alpha(U_\cN)<+\infty\right\}
\end{align}
be the set of non-infinite entropy formulas .

\section{Quantum Entropy Inequalities}
\label{sec:Entropy Inequalities}
All known inequalities that constrain entropy allocations in multipartite quantum states can be derived from {\em strong subadditivity} \cite{LiebRuskai}, given by
\begin{align}
I(A;B|C):=H(AC)+H(BC)-H(ABC)-H(C)\geq 0. \label{eq:SSA}
\end{align}
Here $A$, $B$, and $C$ are arbitrary systems. Pippenger distinguished an independent set of basic inequalities on $n$ systems from which all other known inequalities arise as positive linear combinations \cite{Pip03}. These are (1) nonnegativity of entropy $H(A)\geq 0$, (2) strong subadditivity as stated above, (3)weak monotonicity $H(C|A) + H(C|B)\geq 0$, (4) subadditivity $I(A;B):=H(A)+H(B)-H(AB)\geq 0$ and (5) Araki-Lieb inequality $H(AB)+H(A)-H(B)\geq 0$.

\section{No Auxiliary Variables}
\label{sec:noaux}
There is only one standard decoupling, $\tilde{\phi}_{A_1}=\Tr_{A_2}(\phi_{A_1A_2})$ and $\hat{\phi}_{A_2}=\Tr_{A_1}(\phi_{A_1A_2})$, when there are no auxiliary variables. We now characterize the cone of uniformly additive linear entropic quantities. By the Minkowski-Weyl theorem, every polyhedron $P$ has a half-space or {\it H-representation} $P=\{x: Ax\leq b\}$ for some real matrix $A$ and vector $b$, and a vertex or {\it V-representation} $P=\mathrm{conv}(v_1,v_2,\dots,v_n)+\mathrm{nonneg}(r_1,r_2,\dots,r_s)$ where $v_1,v_2,\dots,v_n,r_1,r_2,\dots,r_s$ are real vectors, $\mathrm{conv}$ denotes the convex hull, and $\mathrm{nonneg}$ denotes non-negative linear combinations.

{\em Sufficient conditions:} The quantity
\begin{align}
f_{\alpha}(U_{\cN},\phi_A) = \lambda_1 H(B) + \lambda_2 H(E) + \lambda_3 H(B|E) + \lambda_4 H(E|B)
\end{align}
is uniformly subadditive for all $\lambda_i \geq 0$. To see this, note first that Eq.~$(\ref{eq:SSA})$ implies that $H(B_1B_2) \leq H(B_1)+ H(B_2)$
and $H(B_1B_2|E_1E_2)\leq H(B_1|E_1)+H(B_2|E_2)$. The other terms $H(E)$ and $H(E|B)$ are handled similarly.  We can then use Eq.~$(\ref{Eq:PosLamb})$ to show $f_{\alpha}$ is uniformly subadditive for $\lambda_i \geq 0$. This characterization of the uniform additivity cone is a V-representation where the quantities $H(B)$, $H(E)$, $H(B|E)$, and $H(E|B)$ are a set of extreme rays and the cone contains the origin.

{\em Necessary conditions:} First we express $f_{\alpha}$ in a slightly different way
\begin{align*}
f_{\alpha}(U_{\cN},\phi_A)  &= \lambda_1 H(B) + \lambda_2 H(E) + \lambda_3 H(B|E) + \lambda_4 H(E|B)\\
 & = (\lambda_1-\lambda_4)H(B) + (\lambda_2 - \lambda_3)H(E) + (\lambda_3+\lambda_4)H(BE),
\end{align*}
so that we have
\begin{align*}
f_{\alpha}(U_{\cN},\phi_A) = \alpha_B H(B) + \alpha_E H(E) + \alpha_{BE}H(BE)
\end{align*}
with $\alpha_B = \lambda_1-\lambda_4$, $\alpha_E = \lambda_2 - \lambda_3$, and $\alpha_{BE} = \lambda_3+\lambda_4$. The requirement that $\lambda_i \geq 0$ translates to the conditions
\begin{align}\label{eq:zeroauxhrep}
\alpha_B + \alpha_{BE} &\geq 0 \nonumber \\
\alpha_E + \alpha_{BE} &\geq 0 \nonumber \\
\alpha_B + \alpha_E + \alpha_{BE} &\geq 0 \nonumber \\
\alpha_{BE} &\geq 0.
\end{align}
This characterization of the uniform additivity cone is an H-representation where each inequality corresponds to a face of the cone.

Now we show that these are necessary for uniform subadditivity. To see this, compute
\begin{equation}\label{eq:deltanoaux}
\Delta(f_\alpha, p) = \alpha_B I(B_1;B_2) + \alpha_{E} I(E_1;E_2) + \alpha_{BE} I(B_1E_1; B_2E_2)
\end{equation}
where $p$ denotes a classical distribution on $B_1B_2E_1E_2$ corresponding to the channel output state. We will show that Eq.~$(\ref{eq:zeroauxhrep})$ are necessary by exhibiting distributions $p$ that lead to a negative value of $\Delta(f_\alpha,p)$ when any of the inequalities is violated.

First, suppose $\alpha_B + \alpha_{BE} <0$. Then, by choosing classical probability distribution $p$ such that $E_1 = E_2 =0$ and $B_1 =B_2 =R_1$, with $R_1$ a uniform random bit, we find $\Delta(f_\alpha,p) = \alpha_B + \alpha_{BE} < 0$. We can show $\alpha_E + \alpha_{BE} \geq 0$ is necessary for uniform subadditivity in a similar way. Now, supposing $\alpha_B + \alpha_E + \alpha_{BE}< 0$, we let $B_1=B_2=E_1=E_2 = R_1$ with $R_1$ a random uniform bit and find $\Delta(f_\alpha,p) = \alpha_B + \alpha_E + \alpha_{BE} < 0$. Finally, if $\alpha_{BE}<0$, we can let $B_1=R_1$, $B_2 = R_2$, $E_1 = R_1\oplus R_3$, $E_2 = R_2 \oplus R_3$ with $R_i$ independent random uniform bits.  In this case we find $\Delta(f_\alpha,p) = \alpha_{BE} < 0$.

\section{One Auxiliary Variable}
\label{sec:oneaux}
For one auxiliary variable $V$, there are several choices of standard decouplings taking a state $\phi_{VA_1A_2}$ to states $\tilde{\phi}_{\tilde{V}A_1}$ and $\hat{\phi}_{\hat{V}A_2}$. We define standard decouplings to have $\tilde{V} = \tilde{M}_2 V$ and $\hat{V} = \hat{M}_1 V$ where $\hat{M}_1$ is a collection of output systems from $\cN_1$ and $\tilde{M}_2$ is a collection of output systems from $\cN_2$. Associate integer labels to each collection according to $0,1,2,3 \leftrightarrow \varnothing, B, E, BE$. The standard decouplings are given by an ordered pair of integers $(a,b)$ where $a$ gives $\hat{M}_1$ and $b$ gives $\tilde{M}_2$. Table~\ref{tab:stddecouple} lists the inequivalent standard decouplings.

\begin{table}
\centering
\begin{tabular}{l|c|c|c|l}
case & (a,b) & $\hat{M}_1$ & $\tilde{M}_2$ & equivalents\\
\hline\hline
1. & (3,3) & $B_1E_1$ & $B_2E_2$ & none\\
\hline
2. & (3,1) & $B_1E_1$ & $B_2$ & (1,3), (3,2), (2,3)\\
\hline
3. & (3,0) & $B_1E_1$ & $\varnothing$ & (0,3)\\
\hline
4. & (1,1) &  $B_1$ & $B_2$ & (2,2)\\
\hline
5. & (1,2) & $B_1$ & $E_2$ & (2,1)\\
\hline
6. & (1,0) & $B_1$ & $\varnothing$ & (2,0), (0,1),(0,2)\\
\hline
7. & (0,0) & $\varnothing$ & $\varnothing$ & none
\end{tabular}
\caption{The inequivalent standard decouplings for one auxiliary variable. \label{tab:stddecouple}}
\end{table}

For one auxiliary variable,
\begin{align*}
\Delta^{(a,b)}(f_\alpha,\rho) & = \alpha_{B}I(B_1;B_2) + \alpha_{E}I(E_1;E_2) + \alpha_{BE}I(B_1E_1;B_2E_2) \\
& + \alpha_{V}(H(\tilde{V})+H(\hat{V})-H(V)) \\
& + \alpha_{BV}(H(B_1\tilde{V})+H(B_2\hat{V})-H(B_1B_2V)) \\
& + \alpha_{EV}(H(E_1\tilde{V})+H(E_2\hat{V})-H(E_1E_2V)) \\
& + \alpha_{BEV}(H(B_1E_1\tilde{V})+H(B_2E_2\hat{V})-H(B_1B_2E_1E_2V))
\end{align*}
can be rewritten as
\begin{equation}\label{eq:shorthandnotation}
\Delta^{(a,b)}(\alpha,\rho) = \Delta^\varnothing(\alpha^\varnothing,\rho) + \Delta^{V,(a,b)}(\alpha^V,\rho)
\end{equation}
where we have replaced $f_\alpha$ by the simpler notation $\alpha$ and
\begin{align}
\Delta^{\varnothing}(\alpha^{\varnothing},\rho) & = \alpha_{B}I(B_1;B_2) + \alpha_{E}I(E_1;E_2)  + \alpha_{BE}I(B_1E_1;B_2E_2) \label{eq:deltaempty} \\
\Delta^{V,(a,b)}(\alpha^{V},\rho) & = \left( \sum_{s \in \cP(BE)} \alpha_{sV}\right)H(\hat{M}_1\tilde{M}_2V) + \sum_{s \in \cP(BE)} \alpha_{sV} E_{sV}. \label{eq:deltaVpart}
\end{align}
In these expressions $\rho$ is the state at the channel outputs on which we evaluate the entropic quantities. The $(a,b)$ index labels the different decouplings we may choose, $\alpha^\varnothing = (\alpha_B,\alpha_E,\alpha_{BE})$, $\alpha^V = (\alpha_V,\alpha_{BV},\alpha_{EV},\alpha_{BEV})$, and $\alpha = (\alpha^\varnothing,\alpha^V)$. The first expression $\Delta^{\varnothing}$ is the same as Eq.~$(\ref{eq:deltanoaux})$ in the zero auxiliary case. For each $s$, the term corresponding to $\alpha_{sV}$ in the second expression has the entropic multiple
\begin{align}
E_{sV} = H(s_1\tilde{M}_2V)+ H(\hat{M}_1s_2V) - H(s_1s_2V) - H(\hat{M}_1\tilde{M}_2V). \label{eq:EsV}
\end{align}
If $s = \varnothing$, then Eq.~(\ref{eq:EsV}) takes the value $I(\hat{M_1}; \tilde{M_2}|V)$.  If $s=BE$, it takes the value $I(\hat{M}_1^c; \tilde{M}_2^c| \hat{M}_1\tilde{M}_2V)$ where superscript $c$ denotes the complement in $\{B_j,E_j\}$. The expression is more complicated for other values of $s$.  If $s=B$, it evaluates to expressions given in Table~\ref{tab:EsVB}, and if $s=E$ it evaluates to expressions in Table~\ref{tab:EsVE}.

We now show that the variables $\alpha^\varnothing$ and $\alpha^V$ can be separated and then prove that Figure 5 in the main text characterizes the uniformly additive formulas obtained using standard decouplings.

\subsection{Separation of Variables}

We would now like to show that the $V$-type terms and the $\varnothing$-type terms can be separated. Let
\begin{align}
\Pi^{\varnothing} & = \left \{ \alpha^\varnothing\ |\ \forall \rho,\  \Delta^{\varnothing}(\alpha^{\varnothing},\rho) \geq 0\right \} \cap \cF\nonumber \\
\Pi^{V, (a,b)} & = \left \{ \alpha^V\ |\ \forall \rho,\  \Delta^{V,(a,b)}(\alpha^{V},\rho) \geq 0\right \}\cap\cF\\
\Pi^{(a,b)} & = \left \{ \alpha\ |\ \forall \rho,\  \Delta^\varnothing (\alpha^\varnothing, \rho) + \Delta^{V,(a,b)}(\alpha^{V},\rho) \geq 0\right \}\cap\cF. \nonumber
\end{align}

\begin{lemma}[separation of variables]\label{lem:sep}
Let $\Pi^{\varnothing}$, $\Pi^{V, (a,b)}$, and $\Pi^{(a,b)}$ be as above.  Then
\begin{align*}
\Pi^{(a,b)} & = \{ (\alpha^\varnothing, \alpha^V)\ |\ \alpha^\varnothing \in \Pi^{\varnothing}\ \textrm{and}\ \alpha^V \in \Pi^{V, (a,b)}\}.
\end{align*}
\end{lemma}
\begin{proof}
It is clear that $\Pi^{(a,b)} \supset  \Pi^{\varnothing}  + \Pi^{V, (a,b)}$, since if  $\Delta^{\varnothing}(\alpha^{\varnothing},\rho) \geq 0$ and   $ \Delta^{V,(a,b)}(\alpha^{V},\rho) \geq 0$  then $\Delta^{(a,b)}(\alpha,\rho)\geq 0$. We would also like to show that any $\alpha \in \Pi^{(a,b)}$ can be decomposed as $\alpha = (\alpha^\varnothing, \alpha^V)$ with $\alpha^\varnothing \in\Pi^{\varnothing}$ and $\alpha^V \in \Pi^{V, (a,b)}$.  To this end, let
$\alpha = (\alpha^\varnothing, \alpha^V) \in \Pi^{(a,b)}$.  To begin with, Lemma \ref{Lemma:SigmaZero} tells us that
$\alpha_V + \alpha_{BV}+\alpha_{EV} + \alpha_{BEV} = 0$.
This lets us rewrite Eq.~(\ref{eq:deltaVpart}) as
\begin{align*}
\Delta^{V,(a,b)}(\alpha^V,p) & = \sum_{s \in \cP(BE)} \alpha_{sV} E_{sV}
\end{align*}
with
\begin{align}
E_{sV} & = H(s_1\tilde{M}_2V)+ H(\hat{M}_1s_2V) - H(s_1s_2V) - H(\hat{M}_1\tilde{M}_2V) \nonumber \\
	 & = H(s_1\tilde{M}_2| V)+ H(\hat{M}_1s_2| V) - H(s_1s_2| V) - H(\hat{M}_1\tilde{M}_2| V). \label{Eq:conditionalEsv}
\end{align}

Now, suppose that $\alpha^\varnothing \not\in \Pi^{\varnothing}$.  In that case, as shown in Section~\ref{sec:noaux}, there is a \emph{classical} probability distribution
 $p_1$ on $B_1B_2E_1E_2$ such that $\Delta^\varnothing(\alpha^\varnothing,p_1) < 0$.  However, we can now extend $p_1$ to $V$ by letting $V = (B_1,B_2,E_1,E_2)$ be a perfectly correlated copy of $B_1B_2E_1E_2$.
From Eq.(\ref{Eq:conditionalEsv}) we see that $E_{sV}(p_1)$ is a sum of entropies of subsets of $B_1B_2E_1E_2$ conditioned on $V$, so $E_{sV}(p_1) = 0$ and therefore $\Delta^{V,(a,b)}(\alpha,p_1) = 0$.  But this means that
\begin{align*}
\Delta(\alpha,p_1) = \Delta^\varnothing (\alpha^\varnothing,p_1) +  \Delta^{V, (a,b)} (\alpha^V,p_1) = \Delta^\varnothing (\alpha^\varnothing,p_1) < 0,
\end{align*}
so we must have $\alpha \not\in \Pi^{(a,b)}$ after all.

Now, suppose $\alpha = (\alpha^\varnothing, \alpha^V) \in \Pi^{(a,b)}$, but $\alpha^V \not \in \Pi^{V,(a,b)}$.  This means that there is some $\rho_{B_1E_1B_2E_2V}$ such that $\Delta^{V,(a,b)}(\alpha^V, \rho) < 0$. We use this $\rho$ to define a new state,
\begin{align*}
\sigma_{B_1B_2E_1E_2V^\prime} & = \frac{1}{d_B^4d_E^4}\sum_{i,j}\sum_{k,l} (P_i \otimes P_{j}\otimes P_{k}\otimes P_{l}\otimes I_V) \rho  (P_i^\dagger \otimes P_{j}^\dagger\otimes P_{k}^\dagger\otimes P_{l}^\dagger \otimes I_V) \otimes \proj{i,j,k,l}_{V_1},
\end{align*}
where $i,j=1\dots d_B^2$, $k,l=1\dots d_E^2$, $P_i$ and $P_j$ label the Pauli matrices on $B$, $P_k$ and $P_l$ label the Pauli matrices on $E$, and $V^\prime = VV_1$.  This state is constructed so that
\begin{align*}
\Delta^\varnothing(\alpha^\varnothing,\sigma) & = 0\\
\Delta^{V, (a,b)}(\alpha^V,\sigma) & = \Delta^{V, (a,b)}(\alpha^V,\rho).
\end{align*}
As a result, we also find that
\begin{align*}
\Delta^{(a,b)}(\alpha,\sigma) = \Delta^\varnothing(\alpha^\varnothing,\sigma) + \Delta^{V, (a,b)}(\alpha^V,\sigma) = \Delta^{V, (a,b)}(\alpha^V,\rho)  < 0,
\end{align*}
so that we have $\alpha \not \in \Pi^{(a,b)}$ in this case too.
\end{proof}

\begin{table}
\centering
\begin{tabular}{c|c|l}
$\hat{M}_1$ & $\tilde{M}_2$ & expression \\
\hline\hline
  $\varnothing$ & $\varnothing$ & $I(B_1;B_2|V)$ \\
\hline
$B_1$ & $\varnothing$ & 0 \\
\hline
$E_1$ & $\varnothing$ & $H(B_2|E_1V) - H(B_2|B_1V) = I(B_1;B_2|V) - I(E_1;B_2|V)$  \\
\hline
$B_1E_1$ & $\varnothing$ & $-I(E_1;B_2|B_1V)$ \\
\hline
  $\varnothing$ & $B_2$ & 0 \\
\hline
$B_1$ & $B_2$ & 0 \\
\hline
$E_1$ & $B_2$ & 0 \\
\hline
$B_1E_1$ & $B_2$ & 0 \\
\hline
  $\varnothing$ & $E_2$ & $H(B_1|E_2V) - H(B_1|B_2V) = I(B_1;B_2|V) - I(B_1;E_2|V)$\\
\hline
$B_1$ & $E_2$ & 0 \\
\hline
$E_1$ & $E_2$ & $H(B_1E_2V)+H(E_1B_2V)- H(B_1B_2V)-H(E_1E_2V)$ \\
\hline
$B_1E_1$ & $E_2$ & $I(E_1;E_2|B_1V) - I(E_1;B_2|B_1V) $\\
\hline
  $\varnothing$ & $B_2E_2$ & $-I(B_1;E_2|B_2V)$ \\
\hline
$B_1$ & $B_2E_2$ & 0 \\
\hline
$E_1$ & $B_2E_2$ & $I(E_2;E_1|B_2V) - I(E_2;B_1|B_2V)$ \\
\hline
$B_1E_1$ & $B_2E_2$ & $I(E_1;E_2|B_1B_2V)$
\end{tabular}
\caption{The entropic quantity $E_{sV}$ in Eq.~(\ref{eq:EsV}) evaluates to these expressions when $s=B$.\label{tab:EsVB}}
\end{table}

\begin{table}
\centering
\begin{tabular}{c|c|l}
$\hat{M}_1$ & $\tilde{M}_2$ & expression\\
\hline\hline
  $\varnothing$ & $\varnothing$ & $I(E_1;E_2|V)$ \\
\hline
$E_1$ & $\varnothing$ & 0 \\
\hline
$B_1$ & $\varnothing$ & $H(E_2|B_1V) - H(E_2|E_1V) = I(E_1;E_2|V) - I(B_1;E_2|V)$  \\
\hline
$E_1B_1$ & $\varnothing$ & $-I(B_1;E_2|E_1V)$ \\
\hline
  $\varnothing$ & $E_2$ & 0 \\
\hline
$E_1$ & $E_2$ & 0 \\
\hline
$B_1$ & $E_2$ & 0 \\
\hline
$E_1B_1$ & $E_2$ & 0 \\
\hline
  $\varnothing$ & $B_2$ & $H(E_1|B_2V) - H(E_1|E_2V) = I(E_1;E_2|V) - I(E_1;B_2|V)$\\
\hline
$E_1$ & $B_2$ & 0 \\
\hline
$B_1$ & $B_2$ & $H(E_1B_2V)+H(B_1E_2V)- H(E_1E_2V)-H(B_1B_2V)$ \\
\hline
$E_1B_1$ & $B_2$ & $I(B_1;B_2|E_1V) - I(B_1;E_2|E_1V) $\\
\hline
  $\varnothing$ & $E_2B_2$ & $-I(E_1;B_2|E_2V)$ \\
\hline
$E_1$ & $E_2B_2$ & 0 \\
\hline
$B_1$ & $E_2B_2$ & $I(B_2;B_1|E_2V) - I(B_2;E_1|E_2V)$ \\
\hline
$E_1B_1$ & $E_2B_2$ & $I(B_1;B_2|E_1E_2V)$
\end{tabular}
\caption{The entropic quantity $E_{sV}$ in Eq.~(\ref{eq:EsV}) evaluates to these expressions when $s=E$.\label{tab:EsVE}}
\end{table}

For each standard decoupling, we want to identify parameters $\alpha$ such that $\Delta^{(a,b)}\geq 0$ for all states on systems $B_1B_2E_1E_2V$. We use Lemma~\ref{lem:sep}, and our earlier characterization of the $\alpha^\varnothing$ satisfying $\Delta^{\varnothing}\geq 0$, to separate variables and focus solely on $\Delta^{V,(a,b)}$. Recall also that $\alpha_V + \alpha_{BV}+\alpha_{EV} + \alpha_{BEV} = 0$ for all standard decouplings. In what follows, let $R_1$, $R_2$, and $R_3$ denote independent uniform 0-1 random variables.

\subsection{Case 1: (3,3) decoupling}

Here we have $\hat{M}_1 = B_1E_1$ and $\tilde{M}_2 = B_2E_2$. We want to compute $\Delta^{V,(3,3)} = \sum_{s \in \cP(BE)} \alpha_{sV} E_{sV}$. For $s = \varnothing$ and $s = BE$ we have $E_{V} = I(B_1E_1;B_2E_2|V)$ and $E_{BEV} = 0$. Consulting Table~\ref{tab:EsVB} and \ref{tab:EsVE}, we find $E_{BV} = I(E_1;E_2|B_1B_2V)$ and $E_{EV} = I(B_1;B_2|E_1E_2V)$ so that
\begin{align}
\Delta^{V,(3,3)} & =\alpha_VI(B_1E_1;B_2E_2|V)  + \alpha_{BV}I(E_1;E_2|B_1B_2V) + \alpha_{EV}I(B_1;B_2|E_1E_2V).
\end{align}
We now need necessary and sufficient conditions on $\alpha$ for $\Delta^{V,(3,3)} \geq 0$.

{\em Necessary conditions:} The conditions
\begin{align}
\alpha_V + \alpha_{BV}+\alpha_{EV} & \geq 0 \label{Eq:Case1:cond1}\\
\alpha_V + \alpha_{BV} & \geq 0\label{Eq:Case1:cond2}\\
\alpha_V + \alpha_{EV} & \geq 0\label{Eq:Case1:cond3}\\
\alpha_V &\geq 0\label{Eq:Case1:cond4}
\end{align}
are necessary for positivity of $\Delta^{V,(3,3)}$. To see the necessity of Eq.~(\ref{Eq:Case1:cond1}), choose $B_1=R_1$, $B_2=R_2$, $E_1=R_1\oplus R_3$, $E_2=R_2\oplus R_3$ and $V = 0$. This give us a distribution with $I(B_1E_1; B_2E_2|V) = 1$, $I(B_1; B_2| E_1E_2V) = 1$, and $I(E_1; E_2|B_1B_2V) = 1$, and so we find Eq.~(\ref{Eq:Case1:cond1}). Eq.~(\ref{Eq:Case1:cond2}) can be seen by choosing $B_1=R_1$, $B_2=R_1\oplus R_2$, $E_1=0$, $E_2=R_2$, and $V = 0$ which results in $I(B_1E_1; B_2E_2|V) = 1$, $I(B_1; B_2| E_1E_2V) = 1$, and $I(E_1; E_2|B_1B_2V) = 0$. Eq.~(\ref{Eq:Case1:cond3}) can be seen in a similar fashion. Finally, to see Eq.~(\ref{Eq:Case1:cond4}), let $B_1 = R_1$, $B_2 = 0$, $E_1 =0$ , $E_2 = R_1$, and $V = 0$.

{\em Sufficient conditions:} We will now show that the necessary conditions for positivity of $\Delta^{V,(3,3)}$ are also sufficient. There are four cases to consider. Suppose first that $\alpha_V,\alpha_{BV},\alpha_{EV} \geq 0$. The positivity of conditional mutual information, and thus the relevant $E_{sV}$'s, makes $\Delta^{V,(3,3)}\geq 0$ immediately. Suppose next that $\alpha_{BV}\geq 0$ and $\alpha_{EV}<0$. In this case, we have
\begin{align*}
\Delta^{V,(3,3)} & =\alpha_VI(B_1E_1;B_2E_2|V)  + \alpha_{BV}I(E_1;E_2|B_1B_2V) + \alpha_{EV}I(B_1;B_2|E_1E_2V)\\
& = (\alpha_V+\alpha_{EV})I(B_1E_1;B_2E_2|V) + \alpha_{BV}I(E_1;E_2|B_1B_2V) + |\alpha_{EV}|\left[ I(B_1E_1;B_2E_2|V) - I(B_1;B_2|E_1E_2V)\right]\\
& = (\alpha_V+\alpha_{EV})I(B_1E_1;B_2E_2|V) + \alpha_{BV}I(E_1;E_2|B_1B_2V) + |\alpha_{EV}|\left[ I(E_1;B_2|E_2V) + I(B_1E_1;B_2|V)\right] \geq 0,
\end{align*}
where we have used
\begin{align*}
I(B_1;B_2|E_1E_2V) & = I(B_1E_1;B_2E_2|V) - I(E_1;B_2|E_2V) - I(B_1E_1;B_2|V)
\end{align*}
and the positivity of conditional mutual information. The case $\alpha_{BV}<0$ and $\alpha_{EV}\geq 0$ follows the same argument as the second case. Finally suppose that $\alpha_{BV}<0$ and $\alpha_{EV}<0$. In this case we have
\begin{align*}
\Delta^{V,(3,3)} & =\alpha_VI(B_1E_1;B_2E_2|V)  + \alpha_{BV}I(E_1;E_2|B_1B_2V) + \alpha_{EV}I(B_1;B_2|E_1E_2V)\\
& = (\alpha_V+\alpha_{BV} + \alpha_{EV})I(B_1E_1;B_2E_2|V) + |\alpha_{BV}|\left[I(B_1E_1;B_2E_2|V) -I(E_1;E_2|B_1B_2V)\right] \\
& \ \ + |\alpha_{EV}|\left[ I(B_1E_1;B_2E_2|V) - I(B_1;B_2|E_1E_2V)\right]\\
& = (\alpha_V+\alpha_{BV}+\alpha_{EV})I(B_1E_1;B_2E_2|V) + |\alpha_{BV}|\left[  I(B_1;E_2|B_2V) +  I(B_1E_1;E_2|V)  \right] \\
& \ \ + |\alpha_{EV}|\left[ I(E_1;B_2|E_2V) + I(B_1E_1;B_2|V)\right] \geq 0.
\end{align*}

\subsection{Case 2: (3,1) decoupling}

Here we have $\hat{M}_1 = B_1E_1$  and $\tilde{M}_2 = B_2$.
For $s = \varnothing$ and $s = BE$ we have $E_{V} = I(B_1E_1;B_2|V)$ and
$E_{BEV} = 0$, respectively, while for $s = B$ and $s=E$, we find $E_{BV} = 0$ and $E_{EV} = I(B_1;B_2|E_1V) - I(B_1;E_2|E_1V)$ from Table~\ref{tab:EsVB} and \ref{tab:EsVE}. This gives us
\begin{equation}
\Delta^{V,(3,1)} = \alpha_VI(B_1E_1;B_2|V)  + \alpha_{EV}\left[ I(B_1;B_2|E_1V) - I(B_1;E_2|E_1V) \right].
\end{equation}

{\em Necessary conditions:} We wish to show that in order to have $\Delta^{V,(3,1)} \geq 0$ for all distributions, we need
\begin{align}
\alpha_{EV} &\leq 0\label{Eq:Case2:cond2}\\
\alpha_{V}+\alpha_{EV} &\geq 0 \label{Eq:Case2:cond3}.
\end{align}
To see that Eq.~(\ref{Eq:Case2:cond2}) is true, choose $B_1 =E_2 = R_1$, and $E_1=B_2 =V = 0$ to get $\Delta^{V,(3,1)} = -\alpha_{EV}$, so that $\alpha_{EV} \leq 0$. To see that Eq.~(\ref{Eq:Case2:cond3}) is necessary, choose $E_2=E_1=V=0$ and $B_1=B_2=R_1$.  Then $\Delta^{V,(3,1)}=\alpha_V+\alpha_{EV}$ which means we need $\alpha_V+\alpha_{EV}\geq 0$.

{\em Sufficient conditions:} Let $\alpha_{EV} \leq 0$ and $\alpha_V + \alpha_{EV} \geq 0$.  Then
\begin{align*}
\Delta^{V,(3,1)} & = \alpha_VI(B_1E_1;B_2|V)  + \alpha_{EV}\left[ I(B_1;B_2|E_1V) - I(B_1;E_2|E_1V) \right]\\
& =  (\alpha_V - |\alpha_{EV}|)I(B_1E_1;B_2|V)  + |\alpha_{EV}|\left[  I(B_1E_1;B_2|V)- I(B_1;B_2|E_1V)\right]  + |\alpha_{EV}| I(B_1;E_2|E_1V) \\
& =  (\alpha_V + \alpha_{EV})I(B_1E_1;B_2|V) + |\alpha_{EV}|\left[  I(E_1;B_2|V)\right] + |\alpha_{EV}| I(B_1;E_2|E_1V) \geq 0.
\end{align*}

\subsection{Case 3: (3,0) decoupling}

Here we have $\hat{M}_1 = B_1E_1$ and $\tilde{M}_2 = \varnothing$ which leads to $E_V = 0$, $E_{BEV} = 0$, $E_{BV} = -I(E_1;B_2|B_1V)$, and $E_{EV} = -I(B_1;E_2| E_1V)$. Therefore,
\begin{equation}
\Delta^{V,(3,0)} = -\alpha_{BV}I(E_1;B_2|B_1V)  -\alpha_{EV}I(B_1;E_2|E_1V). \label{eq:30dec}
\end{equation}

{\em Necessary conditions:} We will need to have
\begin{align}
\alpha_{BV}&\leq 0 \label{eq:30_1}\\
\alpha_{EV} &\leq 0. \label{eq:30_2}
\end{align}
To see Eq.~(\ref{eq:30_1}), choose $E_1 = B_2 = R_1$ and $E_2=B_1=V=0$ to get $\Delta^{V,(3,0)} = -\alpha_{BV} \geq 0$. Similarly, choosing $B_1=E_2=R_1$ and $E_1=B_2=V=0$ gives $\Delta^{V,(3,0)} = -\alpha_{EV} \geq 0$ and Eq.~(\ref{eq:30_2}).

{\em Sufficient conditions:} Eq.(\ref{eq:30dec}) is explicitly nonnegative when $\alpha_{BV}\leq 0$ and $\alpha_{EV} \leq 0$.

\subsection{Case 4: (1,1) decoupling}

Here we have $\hat{M}_1 = B_1$ and $\tilde{M}_2 = B_2$, which gives
$E_V = I(B_1;B_2| V)$, $E_{BEV} = I(E_1;E_2| B_1B_2V)$, $E_{BV} = 0$, and
$E_{EV} =  H(E_1B_2V) + H(B_1E_2V) - H(B_1B_2V) - H(E_1E_2V)$. Therefore
\begin{align}
\Delta^{V,(1,1)} & = \alpha_V I(B_1;B_2| V) + \alpha_{EV}\left[H(E_1B_2V) + H(B_1E_2V) - H(B_1B_2V) - H(E_1E_2V) \right] \\
& + \alpha_{BEV}I(E_1;E_2| B_1B_2V).
\end{align}

{\em Necessary conditions:} We need to have
\begin{align}
\alpha_{EV} &= 0 \label{eq:11_1}\\
\alpha_V &\geq 0 \label{eq:11_2} \\
\alpha_{BEV}&\geq 0. \label{eq:11_3}
\end{align}
Choosing $B_1 = E_2 = R_1$ and $E_1 = B_2 =V=0$, we find $-\alpha_{EV} \geq 0$ so that $\alpha_{EV}\leq 0$. Choosing $B_1 = R_1$, $B_2 = R_2$, $E_1 =E_2= R_1 \oplus R_2$, and $V=0$, we find $\Delta^{V,(1,1)} = \alpha_{EV} \geq 0$, so that $\alpha_{EV} = 0$, showing Eq.(\ref{eq:11_1}). Thus, we have
\begin{equation*}
\Delta^{V,(1,1)} = \alpha_V I(B_1;B_2| V) + \alpha_{BEV}I(E_1;E_2| B_1B_2V),
\end{equation*}
from which we see Eq.~(\ref{eq:11_2}) and Eq.~(\ref{eq:11_3}).

{\em Sufficient conditions:} The sufficiency of $\alpha_V\geq 0$, $\alpha_{BEV}\geq 0$ and $\alpha_{EV}=0$ is immediate from positivity of conditional mutual information.

\subsection{Case 5: (1,2) decoupling}

Here we have $\hat{M}_1 = B_1$ and $\tilde{M}_2 = E_2$, which gives $E_V = I(B_1;E_2|V)$, $E_{BEV} = I(E_1;B_2|B_1E_2V)$, $E_{BV} = 0$, and $E_{EV} = 0$. This leads to
\begin{equation*}
\Delta^{V,(1,1)} = \alpha_V I(B_1;E_2|V) + \alpha_{BEV} I(E_1;B_2|B_1E_2V).
\end{equation*}

{\em Necessary conditions: } We need to have
\begin{align}
\alpha_V &\geq 0 \label{eq:12_1} \\
\alpha_{BEV} &\geq 0. \label{eq:12_2}
\end{align}
Choosing $B_1 = E_2 =V =0$, and $E_1 = B_2 = R_1$, we get Eq.(\ref{eq:12_2}).
Letting $B_1 = B_2 =E_1=E_2=R_1$ and $V=0$ we get Eq.(\ref{eq:12_1}).

{\em Sufficient conditions: } Sufficiency is immediate from positivity of conditional mutual information.

\section{Multiple Auxiliary Variables}
\label{sec:multiple-variables}
We now consider the general case with multiple auxiliary variables $V_1,\ldots,V_n$. We will prove that we can separate the variables, similar to the one-variable case. As a result, under a standard decoupling, the cone of uniformly additive entropic formulas is decomposed into a sum of smaller cones, each of which involves one specific subset of the auxiliary variables. Furthermore, the characterizations of these smaller cones is identical with the ones for zero and one auxiliary variable, which we have given in the previous sections. This will finish the characterization of the additive cone under standard decouplings.

Let $\rho_{V_1\ldots V_nB_1E_1B_2E_2}=(I\otimes U_{\cN_1} \otimes U_{\cN_2})\phi_{V_1\ldots V_nA_1A_2}(I\otimes U_{\cN_1}^\dagger \otimes U_{\cN_2}^\dagger)$ be a state generated by the channels $\cN_1$ and $\cN_2$. We are considering entropic quantities evaluated on systems $V_1...V_nB_1E_1B_2E_2$.  A standard decoupling is an assignment $\tilde{V}_i =  N^i_2 V_i$, $\hat{V}_i = N^i_1 V_i$, where $N^i_1$ is picked from $\cP(B_1E_1)$ and $N^i_2$ is picked from $\cP(B_2E_2)$. We require the decoupling to be consistent: each of $B_2$ and $E_2$ appears in at most one $N_2^i$ and each of $B_1$ and $E_1$ appears in at most one $N_1^i$, such that the new auxiliary variables have no overlaps. A consistent standard decoupling will be indexed by $(a_1,b_1)...(a_n,b_n)$.

Let $J\subseteq [n]:=\{1,2,\dots,n\}$ be a set of indices and $V_J$ denote the collection of systems $V_1\dots V_n$ indexed by $J$. Likewise let $N_1^J$ and $N_2^J$ denote collections of systems $\{N_1^i\}$ and $\{N_2^i\}$ respectively. Note that $V_\varnothing=\varnothing$, $N_1^\varnothing=\varnothing$ and $N_2^\varnothing=\varnothing$. For $\alpha$ being the coefficient vector of an entropy formula $f_\alpha$, we let $\alpha^{V_J}=(\alpha_{V_J},\alpha_{BV_J},\alpha_{EV_J},\alpha_
{BEV_J})$. In Lemma~\ref{Lemma:SigmaZero} we found that if $f_\alpha$ is bounded and uniformly additive with respect to a standard decoupling, then for all $J$ it must hold that
\begin{align}
\label{Eq:Sum-V-zero}
\alpha_{V_J}+\alpha_{BV_J}+\alpha_{EV_J}+\alpha_{BEV_J}=0.
\end{align}
So in the following we assume Eq.(\ref{Eq:Sum-V-zero}). Thus, we can write
\begin{align}
\Delta^{V_J,(a_J,b_J)}(\alpha^{V_J},\rho) &= \sum_{s \in \cP(BE)}\alpha_{sV_J}\Big(H(s_1N^{J}_2|V_J) + H(N^{J}_1s_2|V_J)- H(s_1s_2|V_J) - H(N^J_1N^J_2|V_J)\Big), \label{eq:delta-part}\\
\Delta^{(a_1,b_1)...(a_n,b_n)}(\alpha,\rho) &= \sum_{J} \Delta^{V_J,(a_J,b_J)}(\alpha^{V_J},\rho), \label{eq:delta-total}
\end{align}
where $(a_J,b_J)$ tells us which systems from $\{B_1,E_1,B_2,E_2\}$ go with $V_J$ and is induced from $(a_1,b_1)...(a_n,b_n)$. We can now define the cones
\begin{align*}
\Pi^{V_J,(a_J,b_J)}&=\left\{\alpha^{V_J}|\ \forall \rho,\ \Delta^{V_J,(a_J,b_J)}(\alpha^{V_J},\rho)\geq 0\right\} \cap \cF, \\
\Pi^{(a_1,b_1)...(a_n,b_n)}&=\left\{\alpha|\ \forall \rho,\ \Delta^{(a_1,b_1)...(a_n,b_n)}(\alpha,\rho)\geq 0\right\}\cap \cF.
\end{align*}

If $V_J=\varnothing$, the characterization of $\Pi^{V_J,(a_J,b_J)}$ has been given in Section~\ref{sec:noaux}. Note that in this case $a_J=0$ and $b_J=0$ correspond to empty sets and they are meaningless. If $V_J\neq \varnothing$, we can regard $V_J$ as a single auxiliary variable and find the explicit description of $\Pi^{V_J,(a_J,b_J)}$ in Section~\ref{sec:oneaux}. On the other hand, the cone $\Pi^{(a_1,b_1)...(a_n,b_n)}$ includes all the uniformly additive quantities $f_\alpha$, under the decoupling $(a_1,b_1)...(a_n,b_n)$. Our main result in this section is the following Theorem~\ref{thm:separation}, which gives a simple characterization of $\Pi^{(a_1,b_1)...(a_n,b_n)}$, in terms of $\Pi^{V_J,(a_J,b_J)}$.

\begin{theorem}
  \label{thm:separation}
Given $\alpha$, we have
\begin{align}
  \label{eq:whole-cone}
  \alpha \in \Pi^{(a_1,b_1)...(a_n,b_n)}
\end{align}
if and only if
\begin{align}
  \label{eq:sub-cone}
  \forall J\subseteq [n],\quad \alpha^{V_J} \in \Pi^{V_J,(a_J,b_J)}.
\end{align}
\end{theorem}

The proof of Theorem~\ref{thm:separation} uses the following two lemmas.

\begin{lemma}\label{Lemma:classical-distribution}
Let $\Pi^{V_J, (a_J,b_J)}$ be defined as above.  Then, if $\alpha^{V_J} \not \in \Pi^{V_J,(a_J,b_J)}$ there is a classical probability distribution $p$ on $B_1E_1B_2E_2V_J$ such that $\Delta^{V_J, (a_J,b_J)}(\alpha^{V_J},p) <0$.
\end{lemma}

\begin{proof}
This is shown in Section~\ref{sec:noaux} and Section~\ref{sec:oneaux}.
\end{proof}

\begin{lemma} \label{Lemma:distribution-transform}
Fix a probability distribution $p$ on $V_1\dots V_n B_1E_1B_2E_2$ and a consistent standard decoupling $(a_1,b_1)...(a_n,b_n)$. Let $T \subseteq [n]$ be a fixed set and $(a_T,b_T)$ be the induced standard decoupling associated with the set of variables $V_T$. Then we can construct a probability distribution $p^\prime$ on ${V}_1^\prime\dots {V}_n^\prime B_1^\prime E_1^\prime B_2^\prime E_2^\prime$ such that
\begin{align}\label{eq:distri-tran}
\Delta^{(a_1,b_1)...(a_n,b_n)}(\alpha,p^\prime) = \Delta^{V_T,(a_T,b_T)}(\alpha^{V_T},p).
\end{align}
\end{lemma}

\begin{proof}
If the systems $B_1,B_2,E_1,E_2$ do not have the same size, we extend them such that their sizes are the same. Denote $d=|B_1|=|B_2|=|E_1|=|E_2|$. We let $k = |T|$, the number of indices in $T$, and let $t_i$, $i=1,\dots,k$, be the $i$th element of $T$. For each $f=0,\dots,3$ and $i=1,\dots,k$, choose $R^f_i$ to be an independent uniformly distributed variable on $\{0,1,\ldots,d-1\}$. Let $R^f = \sum_{i=1}^k R^f_{i}\ (\mathrm{mod}\ d)$. To define $p^\prime$, we let
\begin{align}
B_1^\prime & = R^0 + B_1\ (\mathrm{mod}\ d) & E_1^\prime & = R^1 + E_1\ (\mathrm{mod}\ d)\\
B_2^\prime & = R^2 + B_2\ (\mathrm{mod}\ d) & E_2^\prime & = R^3 + E_2\ (\mathrm{mod}\ d).
\end{align}
For $i=1,\dots,k$, we let $V^\prime_{t_i} = V_{t_i}R^0_i R^1_i R^2_i R^3_i$, and for $r \not\in T$ we choose $V^\prime_{r} = B_1E_1B_2E_2$.

For any $X\in\cP(B_1E_1)$ and $Y\in\cP(B_2E_2)$, we let $X^\prime$ and $Y^\prime$ be the corresponding collections of systems from $B_1^\prime E_1^\prime$ and $B_2^\prime E_2^\prime$, respectively (i.e., if $X=B_1E_1$, then $X^\prime=B_1^\prime E_1^\prime$). Since $V_T^\prime$ includes $V_T$, as well as all the $R_i^f$ variables from which we know $R^0, R^1, R^2$ and $R^3$, we have
\[
  H(X^\prime Y^\prime|V_T^\prime)=H(X Y|V_T).
\]
This, combined with Eq.~(\ref{eq:delta-part}), gives
\begin{align}\label{eq:distri-tran-1}
\Delta^{V_T,(a_T,b_T)}(\alpha^{V_T},p^\prime) = \Delta^{V_T,(a_T,b_T)}(\alpha^{V_T},p).
\end{align}
Next, we show that
\begin{align}\label{eq:distri-tran-2}
\Delta^{V_J,(a_J,b_J)}(\alpha^{V_J},p^\prime) = 0, \quad\  \text{for all}\ J\neq T.
\end{align}
For this, we consider two cases. If $T\subset J$, then $B_1^\prime, E_1^\prime, B_2^\prime, E_2^\prime$ are all known given $V_J^\prime$, because $V_J^\prime$ includes $B_1E_1B_2E_2R^0R^1R^2R^3$. As a result,
\[
  H(X^\prime Y^\prime|V_J^\prime)=0.
\]
On the other hand, if $T\not\subseteq J$, There must exist $i$, such that none of $R_i^0, R_i^1, R_i^2, R_i^3$ is included in $V_J^\prime$. Thus given $V_J^\prime$, the variables $R^0, R^1, R^2, R^3$ are independent and uniformly distributed, and so are $B_1^\prime, E_1^\prime, B_2^\prime, E_2^\prime$. As a result,
\[
  H(X^\prime Y^\prime|V_J^\prime)=H(X^\prime)+H(Y^\prime).
\]
In both cases, using Eq.~(\ref{eq:delta-part}) we obtain Eq.~(\ref{eq:distri-tran-2}). At last, using Eq.~(\ref{eq:delta-total}) we easily see that Eq.~(\ref{eq:distri-tran-1}) and Eq.~(\ref{eq:distri-tran-2}) together lead to Eq.~(\ref{eq:distri-tran}).
\end{proof}

\medskip
\begin{proof}[Proof of Theorem \ref{thm:separation}]
It is obvious that Eq.~(\ref{eq:sub-cone}) implies Eq.~(\ref{eq:whole-cone}). For the other direction, we suppose that Eq.~(\ref{eq:sub-cone}) is not true: there is a subset $T\subseteq [n]$, such that $\alpha^{V_T}\not\in  \Pi^{V_T,(a_T,b_T)}$. Then by Lemma~\ref{Lemma:classical-distribution}, there is a probability distribution $p$ on $B_1E_1B_2E_2V_T$, satisfying
\[\Delta^{V_T, (a_T,b_T)}(\alpha^{V_T},p) <0.\]
Due to Lemma~\ref{Lemma:distribution-transform}, this further implies that we have probability distribution $p^\prime$ such that
\[\Delta^{(a_1,b_1)...(a_n,b_n)}(\alpha,p^\prime)<0,\]
which indicates that $\alpha \not\in \Pi^{(a_1,b_1)...(a_n,b_n)}$.
\end{proof}

\section{Non-standard Decouplings}
\label{sec:non-standard-decouplings}
The motivation of our consideration of standard decouplings comes from the experience in proving additivity of certain well-known quantities. However, a general treatment should consider all possible ways to generate the new auxiliary variables in the decoupling. In this section, we investigate the usefulness of non-standard decouplings. Interestingly, we find that all uniformly additive quantities $f_\alpha(U_\cN)$ derived from consistent decouplings that are non-standard (cf. definitions in Section~\ref{sec:General Considerations}), can be obtained by using standard decouplings. This proves that standard decouplings are really typical.

\begin{theorem}
  \label{thm:nonstandard-decoupling}
Let the linear entropy formula $f_\alpha(U_\cN,\phi_{V_1\dots V_nA})$ be bounded and uniformly subadditive with respect to a non-standard, consistent decoupling. Then there is $f_\beta(U_\cN,\varphi_{V_1\dots V_mA})$ defined on states with $m\leq n$ auxiliary variables, such that $f_\beta(U_\cN,\varphi_{V_1\dots V_mA})$ is uniformly subadditive with respect to a standard decoupling and
\[
  \max_{\phi_{V_1\dots V_nA}} f_\alpha(U_\cN,\phi_{V_1\dots V_nA})=\max_{\varphi_{V_1\dots V_mA}}
       f_\beta(U_\cN,\varphi_{V_1\dots V_mA}).
\]
\end{theorem}

Theorem~\ref{thm:nonstandard-decoupling} guarantees that there is no need to find out the uniformly subadditive entropy formulas $f_\alpha(U_\cN,\phi_{V_1\dots V_n A})$ under non-standard consistent decouplings. This is because our interest is in searching for uniformly additive quantities $f_\alpha(U_\cN):=\max_\phi f_\alpha(U_\cN,\phi_{V_1\dots V_n A})$, other than in the entropy formulas themselves. For this purpose, Theorem~\ref{thm:nonstandard-decoupling} shows that our consideration of standard decouplings suffices.

Before going to the proof, we specify some of the notations. Since the linear entropy formula $f_\alpha(U_\cN,\phi_{V_1\dots V_nA})$ is defined with respect to the state $\rho_{V_1\dots V_nBE}=(I \otimes U_{\cN})\phi_{V_1\dots V_nA}(I \otimes U_{\cN}^\dagger)$, we also denote $f_\alpha(U_\cN,\phi_{V_1\dots V_nA})$ as
\begin{equation}
  \label{eq:nonstandard-1}
  f_\alpha(\rho_{V_1\dots V_nBE}):=\sum_{t\in\cP(V_1\dots V_n)}\sum_{s\in\cP(BE)}\alpha_{s,t}H(st)_\rho.
\end{equation}
When non-standard decouplings are considered, we may encounter the situation that some of the auxiliary variables are empties. Let the state $\sigma_{V_1\dots V_nBE}$ have empty auxiliary variables, say, we suppose $V_n=\varnothing$. Then $f_\alpha(\sigma_{V_1\dots V_nBE})$ is evaluated according to Eq.~(\ref{eq:nonstandard-1}) by letting $H(V_n)_\sigma=0$ and $H(MV_n)_\sigma=H(M)_\sigma$ for any $M\in\cP(V_1\dots V_{n-1}BE)$. Such a state $\sigma_{V_1\dots V_nBE}$ with $V_n=\varnothing$ is not artifical: we can identify it in a natural way with $\sigma_{V_1\dots V_{n-1}BE}\otimes\proj{0}_{V_n}$, that is, empty variables are actually each in a pure state and are hence isolated from the other ones.

Let $\cD: \rho_{V_1\dots V_n B_1 B_2 E_1E_2} \rightarrow (\rho_{\tilde{V}_1\dots \tilde{V}_n B_1E_1},\rho_{\hat{V}_1\dots \hat{V}_n B_2E_2})$ be the non-standard decoupling in the assumption of Theorem~\ref{thm:nonstandard-decoupling}. It is determined by a grouping and relabeling of the systems $V_1,\dots, V_n, B_2, E_2$ to form $\tilde{V}_1,\dots, \tilde{V}_n$,
and another grouping and relabeling of the systems $V_1,\dots, V_n, B_1, E_1$ to form $\hat{V}_1,\dots, \hat{V}_n$. That is, $\tilde{V}_i \in \cP(V_1\dots V_n B_2E_2)$ and  $\hat{V}_i \in \cP(V_1\dots V_n B_1E_1)$, and as a consistence condition we require $\tilde{V}_i \cap \tilde{V}_j = \hat{V}_i \cap \hat{V}_j = \varnothing$. We further write $\tilde{V}_i$ and $\hat{V}_i$ as the joint of the ``$V$'' part and the ``$BE$'' part: $\tilde{V}_i=V_i^\prime N_2^i$ with $V_i^\prime\in\cP(V_1\dots V_n)$ and $N_2^i \in \cP(B_2E_2)$, $\hat{V}_i=V_i^{\prime\prime} N_1^i$ with $V_i^{\prime\prime}\in\cP(V_1\dots V_n)$ and $N_1^i \in \cP(B_1E_1)$. In this section, the notations $\tilde{V}_i,\hat{V}_i,V_i^\prime,V_i^{\prime\prime},N_1^i,N_2^i$ with $i=1,\dots,n$ are all reserved to denote the fixed sets of variables given by the decoupling $\cD$, as described above.

\begin{definition}
  \label{def:relocation-rule}
  Given the sets $T_1,\dots,T_n\in [n]$ such that $T_i\cap T_j = \varnothing$ for $1\leq i \neq j \leq n$, we define a relocation rule $g$ of the variables $W_1,\dots, W_n$, via
  \[
    g(W_1,\dots, W_n):=(W_{T_1},\dots, W_{T_n}),
  \]
  where $W_{T_i}$ is a collection of the systems $W_j$ such that $j\in T_i$.
\end{definition}

According to Definition~\ref{def:relocation-rule}, we now define two relocation rules $g_1$ and $g_2$, which are associated with the decoupling $\cD$ and satisfy
\begin{align*}
g_1(V_1,\dots,V_n)&=(V_1^\prime,\dots,V_n^\prime), \\
g_2(V_1,\dots,V_n)&=(V_1^{\prime\prime},\dots,V_n^{\prime\prime}).
\end{align*}
That is, $g_1$ is given by the sets $T_i:=\{j|\ 1\leq j\leq n, V_j\in V_i^\prime\}$ with $i=1,\dots,n$, and $g_2$ is given by the sets $S_i:=\{j|\ 1\leq j\leq n, V_j\in V_i^{\prime\prime}\}$ with $i=1,\dots,n$.

The following lemma will be very useful. Note that in Eqs.~(\ref{eq:lemma-nonstandard-1}) and~(\ref{eq:lemma-nonstandard-2}), $V_i^\prime$ and $V_i^{\prime\prime}$ are actually collections of the variables $V_1,\dots,V_n$, formulated by the relocation rules $g_1$ and $g_2$. So in later applications of Lemma~\ref{lemma:nonstandard}, we may also use $g_1$ and $g_2$ to specify the relations between the auxiliary variables.
\begin{lemma}
  \label{lemma:nonstandard}
  Under the same assumption of Theorem~\ref{thm:nonstandard-decoupling} and using the notations described above, we have for any state $\rho_{V_1\dots V_nBE}$,
  \begin{align}
    f_\alpha(\rho_{V_1\dots V_nBE}) &\leq f_\alpha(\rho_{V_1^\prime\dots V_n^\prime BE}),  \label{eq:lemma-nonstandard-1} \\
    f_\alpha(\rho_{V_1\dots V_nBE}) &\leq f_\alpha(\rho_{V_1^{\prime\prime}\dots V_n^{\prime\prime}
                                                                                      BE}).\label{eq:lemma-nonstandard-2}
  \end{align}
\end{lemma}

\begin{proof}
At first, it has been shown in Lemma~\ref{Lemma:SigmaZero} (Eq.~(\ref{Eq:SumLEQ})) that $f_\alpha(U_\cN,\phi_{V_1\dots V_nA})$ being bounded implies that
\begin{align}
  \label{eq:lemma-nonstandard-3}
  f_\alpha(\rho_{V_1\dots V_n}\otimes\proj{00}_{BE})=\sum_{t\in\cP(V_1\dots V_n)}\sum_{s\in\cP(BE)}
                                                      \alpha_{s,t}H(t)_{\rho_{V_1\dots V_n}} \leq 0
\end{align}
for any state $\rho_{V_1\dots V_n}$. Now since $f_\alpha(U_\cN,\phi_{V_1\dots V_nA})$ is uniformly subadditive with respect to the decoupling $\cD$, we have
\begin{align}
  \label{eq:lemma-nonstandard-4}
  \Delta(\alpha,\rho_{V_1\dots V_nB_1B_2E_1E_2})=f_\alpha(\rho_{\tilde{V}_1\dots \tilde{V}_n B_1E_1})
           +f_\alpha(\rho_{\hat{V}_1\dots \hat{V}_n B_2E_2})-f_\alpha(\rho_{V_1\dots V_n B_1 B_2 E_1E_2})\geq 0
\end{align}
for any state $\rho_{V_1\dots V_n B_1 B_2 E_1E_2}$. Considering a state of the form $\rho_{V_1\dots V_n B_1 E_1}\otimes\rho_{B_2 E_2}$, we derive from Eq.~(\ref{eq:lemma-nonstandard-4}) that
\begin{equation}\begin{split}
  \label{eq:lemma-nonstandard-5}
  \Delta(\alpha,\rho_{V_1\dots V_n B_1 E_1}\otimes\rho_{B_2 E_2})
  = &\sum_{s,t}\alpha_{s,t}\big(H(s_1\tilde{t})+H(s_2\hat{t})-H(s_1s_2t)\big)  \\
  = &\sum_{s,t}\alpha_{s,t}\big(H(s_1t^\prime)+H(\tilde{t}/t^\prime)+H(s_2)+H(\hat{t})-H(s_2)-H(s_1t)\big)  \\
  =&\sum_{s,t}\alpha_{s,t}\big(H(s_1t^\prime)-H(s_1t)\big)+\sum_{s,t}\alpha_{s,t}\big(H(\tilde{t}/t^\prime)+H(\hat{t})\big)\\
  \geq &0,
\end{split}\end{equation}
where the sums are over all subsets $s\in\cP(BE)$ and $t\in\cP(V_1\dots V_n)$, and the notation $\tilde{t}/t^\prime$ indicates the collection of variables resulting from removing $t^\prime$ from $\tilde{t}$. Eq.~(\ref{eq:lemma-nonstandard-3}) gives
\[
  \sum_{s,t}\alpha_{s,t} H(\tilde{t}/t^\prime)\leq 0  \quad\text{and}\quad  \sum_{s,t}\alpha_{s,t} H(\hat{t}) \leq 0.
\]
Combining this with Eq.~(\ref{eq:lemma-nonstandard-5}) we conclude that for any state $\rho_{V_1\dots V_n B_1 E_1}$,
\[
  \sum_{s,t}\alpha_{s,t}\big(H(s_1t^\prime)-H(s_1t)\big)\geq 0,
\]
which proves Eq.~(\ref{eq:lemma-nonstandard-1}). Since Eq.~(\ref{eq:lemma-nonstandard-2}) can be proved in the same way we have finished the proof.
\end{proof}

Now we are ready for the proof of Theorem~\ref{thm:nonstandard-decoupling}. We will not construct an explicit expression for $f_\beta(U_\cN,\varphi_{V_1\dots V_mA})$. Instead, we prove the existence.

\begin{proof}[Proof of Theorem \ref{thm:nonstandard-decoupling}]
We will use mathematical induction. Let us consider the following two cases.

\textbf{Case 1:} $V_i^\prime \neq \varnothing$ and $V_i^{\prime\prime} \neq \varnothing$ for all $1\leq i\leq n$. In this case, $V_1^\prime,\dots, V_n^\prime$ and $V_1^{\prime\prime},\dots, V_n^{\prime\prime}$ are respectively permutations of $V_1,\dots, V_n$: there are permutations $\pi,\tau \in S_n$ such that $V_i^\prime=V_{\pi(i)}$ and $V_i^{\prime\prime}=V_{\tau(i)}$ for all $1\leq i\leq n$. Denote the order of $\pi$ and $\tau$ as $a$ and $b$, respectively. That is
\[
  \pi^a=\tau^b=I,
\]
where $I$ is the identity of the symmetric group $S_n$. Now define
 \begin{align}
   \left(\tilde{V}_1^{(a-1)},\dots, \tilde{V}_n^{(a-1)}\right)&:=g_1^{a-1}\left(\tilde{V}_1,\dots,\tilde{V}_n\right),  \label{eq:nonstandard-proof-1} \\
   \left(\hat{V}_1^{(b-1)},\dots, \hat{V}_n^{(b-1)}\right)&:=g_2^{b-1}\left(\hat{V}_1,\dots,\hat{V}_n\right).
   \label{eq:nonstandard-proof-2}
\end{align}
Then
\begin{align}
  \tilde{V}_i^{(a-1)}&=\tilde{V}_{\pi^{a-1}(i)}=V^\prime_{\pi^{a-1}(i)}N_2^{\pi^{a-1}(i)}=V_{\pi^{a}(i)}N_2^{\pi^{a-1}(i)}
                                              =V_iN_2^{\pi^{a-1}(i)},  \label{eq:nonstandard-proof-3} \\
  \hat{V}_i^{(b-1)}&=\hat{V}_{\tau^{b-1}(i)}=V^{\prime\prime}_{\tau^{b-1}(i)}N_1^{\tau^{b-1}(i)}=V_{\tau^{b}(i)}
                                    N_1^{\tau^{b-1}(i)}=V_iN_1^{\tau^{b-1}(i)}.  \label{eq:nonstandard-proof-4}
\end{align}
To proceed, for any state $\rho_{V_1\dots V_nB_1B_2E_1E_2}$, we have
\begin{equation}\begin{split}
  \label{eq:nonstandard-proof-5}
  f_\alpha(\rho_{V_1\dots V_nB_1B_2E_1E_2}) &\leq f_\alpha(\rho_{\tilde{V}_1\dots \tilde{V}_nB_1E_1})
                                                               +f_\alpha(\rho_{\hat{V}_1\dots \hat{V}_nB_2E_2})  \\
  &\leq f_\alpha\big(\rho_{\tilde{V}_1^{(a-1)}\dots \tilde{V}_n^{(a-1)}B_1E_1}\big)
                                      +f_\alpha\big(\rho_{\hat{V}_1^{(b-1)}\dots \hat{V}_n^{(b-1)}B_2E_2}\big),
\end{split}\end{equation}
where the first inequality is by assumption, and for the second inequality we have applied Lemma~\ref{lemma:nonstandard} iteratively and used the notations defined in Eqs.~(\ref{eq:nonstandard-proof-1}) and (\ref{eq:nonstandard-proof-2}). Eq.~(\ref{eq:nonstandard-proof-5}) shows that $f_\alpha(\rho_{V_1\dots V_nBE})$ itself is uniformly subadditive with respect to a \emph{standard} decoupling given by Eqs.~(\ref{eq:nonstandard-proof-3}) and~(\ref{eq:nonstandard-proof-4}).

\textbf{Case 2:} at least one of $V_i^\prime$ for $i=1,\dots,n$ or one of $V_i^{\prime\prime}$ for $i=1,\dots,n$ is $\varnothing$. Without loss of generality, we suppose $V_i^\prime=\varnothing$ for some values of $i$, and further suppose that all the empty variables are in the end. So there is $k<n$, such that $V_1^\prime\dots V_n^\prime=V_1^\prime\dots V_k^\prime\varnothing\dots\varnothing$ (i.e., $V_i^\prime=\varnothing$ for $i=k+1,\dots,n$). Note that it is possible that $k=0$. Now Eq.~(\ref{eq:lemma-nonstandard-1}) of Lemma~\ref{lemma:nonstandard} translates to
\begin{equation}
  \label{eq:nonstandard-proof-6}
  f_\alpha(\rho_{V_1\dots V_nBE}) \leq f_\alpha\big(\rho_{V_1^\prime\dots V_k^\prime BE}
             \otimes\proj{0}_{V_{k+1}^\prime}\otimes\dots\otimes\proj{0}_{V_n^\prime}\big).
\end{equation}
Define a linear entropy formula $f_\gamma(\rho_{V_1\dots V_kBE})$ on states with $k$ auxiliary variables, as
\begin{equation}
  \label{eq:nonstandard-proof-7}
  f_\gamma(\rho_{V_1\dots V_kBE}):= f_\alpha\big(\rho_{V_1\dots V_k BE}
             \otimes\proj{0}_{V_{k+1}}\otimes\dots\otimes\proj{0}_{V_n}\big).
\end{equation}
We now claim:
\begin{itemize}
  \item[(A)] It holds that
    \[\max_{\phi_{V_1\dots V_nA}} f_\alpha\big(U_\cN(\phi_{V_1\dots V_nA})\big)
    = \max_{\varphi_{V_1\dots V_kA}} f_\gamma\big(U_\cN(\varphi_{V_1\dots V_kA})\big).\]
    In particular, this equality implies that $f_\gamma\big(U_\cN,\varphi_{V_1\dots V_kA}\big)$ is also bounded.
  \item[(B)] $f_\gamma\big(\rho_{V_1\dots V_kBE}\big)$ is uniformly subadditive with respect to a \emph{consistent} decoupling.
\end{itemize}
Claim (A) is easy to see. The ``$\leq$'' part follows from Eq.~(\ref{eq:nonstandard-proof-6}), and the ``$\geq$'' part is obvious by the definition Eq.~(\ref{eq:nonstandard-proof-7}). To verify claim (B), for any state $\rho_{V_1\dots V_kB_1B_2E_1E_2}$ we have
\begin{equation}\begin{split}
  \label{eq:nonstandard-proof-8}
  f_\gamma(\rho_{V_1\dots V_kB_1B_2E_1E_2})
  =    &f_\alpha\big(\rho_{V_1\dots V_k B_1B_2E_1E_2}\otimes\proj{0}_{V_{k+1}}\otimes\dots\otimes\proj{0}_{V_n}\big)\\
  \leq &f_\alpha(\rho_{\tilde{V}_1\dots \tilde{V}_nB_1E_1})+f_\alpha(\rho_{\hat{V}_1\dots \hat{V}_nB_2E_2})   \\
  \leq &f_\alpha\big(\rho_{\tilde{V}_1^\prime\dots \tilde{V}_k^\prime B_1E_1}\!\otimes\!\proj{0}_{\tilde{V}_{k+1}^\prime}
        \!\!\!\otimes\!\dots\!\otimes\!\proj{0}_{\tilde{V}_n^\prime}\big) + f_\alpha\big(\rho_{\hat{V}_1^\prime\dots \hat{V}_k^\prime B_2E_2}\!\otimes\!\proj{0}_{\hat{V}_{k+1}^\prime}\!\!\!\otimes\!\dots\!\otimes\!
                                                                              \proj{0}_{\hat{V}_n^\prime}\big)  \\
  =    &f_\gamma\big(\rho_{\tilde{V}_1^\prime\dots \tilde{V}_k^\prime B_1E_1}\big)
            + f_\gamma\big(\rho_{\hat{V}_1^\prime\dots \hat{V}_k^\prime B_2E_2}\big),
\end{split}\end{equation}
where we have defined $(\tilde{V}_1^\prime,\dots, \tilde{V}_n^\prime):=g_1(\tilde{V}_1,\dots, \tilde{V}_n)$ and $(\hat{V}_1^\prime,\dots, \hat{V}_n^\prime):=g_1(\hat{V}_1,\dots, \hat{V}_n)$, and since the second line, we have set $V_{k+1}=\dots=V_n=\varnothing$. In Eq.~(\ref{eq:nonstandard-proof-8}), the first line is by definition~(\ref{eq:nonstandard-proof-7}), the second line is by assumption that $f_\alpha$ is uniformly subadditive with respect to the decoupling $\cD$, the third line is by Lemma~\ref{lemma:nonstandard} (in the form of Eq.~(\ref{eq:nonstandard-proof-6})), and the last line is again by definition~(\ref{eq:nonstandard-proof-7}). We can check that $\tilde{V}_i^\prime\in\cP(V_1\dots V_kB_2E_2)$ and $\hat{V}_i^\prime\in\cP(V_1\dots V_kB_1E_1)$  for $1\leq i\leq k$,
and also $\tilde{V}_i^\prime\cap\tilde{V}_j^\prime=\varnothing$ and $\hat{V}_i^\prime\cap\hat{V}_j^\prime=\varnothing$ for $i\neq j$. Thus
\begin{align}
  \label{eq:nonstandard-proof-9}
  \rho_{V_1\dots V_kB_1B_2E_1E_2}\rightarrow \big(\rho_{\tilde{V}_1^\prime\dots \tilde{V}_k^\prime B_1E_1},
                                                       \rho_{\hat{V}_1^\prime\dots \hat{V}_k^\prime B_2E_2} \big)
\end{align}
is a consistent decoupling and Eq.~(\ref{eq:nonstandard-proof-8}) indeed verifies that $f_\gamma\big(\rho_{V_1\dots V_kBE}\big)$ is uniformly subadditive with respect to this decoupling.

At last, we argue that the above considerations of Case 1 and Case 2 suffice to conclude the proof, by applying the method of mathematical induction, of which the basis here is the fact that with zero auxiliary variable, the unique consistent decoupling $\rho_{B_1B_2E_1E_2}\rightarrow (\rho_{B_1E_1}, \rho_{B_2E_2})$ is standard. Note that the proofs for the two claims in Case 2 work as well when $k=0$ and in this case the decoupling~(\ref{eq:nonstandard-proof-9}) reduces to the standard decoupling $\rho_{B_1B_2E_1E_2}\rightarrow (\rho_{B_1E_1}, \rho_{B_2E_2})$.
\end{proof}

\section{Entropic Criterion for C-Q Structure of Quantum State}
\label{sec:c-q-structure}
\begin{lemma}
  \label{lemma:entropy-cq-state}
  For a quantum state $\rho_{R_1R_2R_3A}$, suppose the conditional entropies satisfy
  $H(R_i|R_j)=0$ and $H(R_i|R_jR_k)=0$ for all $i,j,k\in\{1,2,3\}$. Then the reduced state
  $\rho_{R_iA}$ is classical-quantum, e.g., $\rho_{R_1A}$  can be written as
  \begin{equation}
    \label{eq:lemma-eccq-1}
    \rho_{R_1A}=\sum_xp_x\proj{x}^{R_1}\otimes\rho_x^A,
  \end{equation}
  with $\{\ket{x}\}$ a set of orthogonal states.
\end{lemma}

\begin{proof}
Since $I(R_1;R_2|R_3)=H(R_1|R_3)-H(R_1|R_2R_3)=0$, by the result of~\cite{SSA-eq}, we know that the reduced state $\rho_{R_1R_2}$ is separable.
Thus we can write
\[
  \rho_{R_1R_2}=\sum_{x=1}^Mp_x\sigma_x^{R_1}\otimes\omega_x^{R_2},
\]
and without loss of generality we assume $\sigma_{x_1}\neq\sigma_{x_2}$ and
$\omega_{x_1}\neq\omega_{x_2}$ for all $1\leq x_1\neq x_2\leq M$. Let
\[
  \rho_{XR_1R_2}=\sum_{x=1}^Mp_x\proj{x}^X\otimes\sigma_x^{R_1}\otimes\omega_x^{R_2}
\]
be an extension of $\rho_{R_1R_2}$, with $\{\ket{x}\}$ a set of orthogonal states. Then we
have
\begin{equation}
  \label{eq:lemma-eccq-2}
  0=H(R_1|R_2)\geq H(R_1|R_2X)=\sum_xp_xH(\sigma_x)\geq 0.
\end{equation}
On the one hand, Eq.~(\ref{eq:lemma-eccq-2}) implies that $\sigma_x$ is pure for all values of $x$.
On the other hand, from Eq.~(\ref{eq:lemma-eccq-2}) we have
\[
  H(R_1|R_2)= H(R_1|R_2X),
\]
which implies that we can recover $\rho_{XR_1R_2}$ from $\rho_{R_1R_2}$ by a CPTP map acting on system
$R_2$ only~\cite{Petz1,Petz2}. This further implies that the set of states $\{\omega_x\}_x$ are
mutually orthogonal. In similar ways, we can show that $\omega_x$ is pure for all values of $x$, and
the set of states $\{\sigma_x\}_x$ are mutually orthogonal. These consequences all together give us that
$\rho_{R_1R_2}$ has the follow form:
\[
  \rho_{R_1R_2}=\sum_{x=1}^Mp_x\proj{x}^{\otimes 2},
\]
with $\{\ket{x}\}$ a set of orthogonal states. This obviously implies Eq.~(\ref{eq:lemma-eccq-1}),
and we are done.
\end{proof}

\section{Informationally Degradable Channels Have Additive Coherent Information}
\label{sec:informational-degradability}
We say that a quantum channel $\cN:A\rightarrow B$ is informationally degradable, if
\[
  I(V;B)\geq I(V;E)
\]
for any state $\rho_{VBE}=(I\otimes U_{\cN})\phi_{VA}(I\otimes U_{\cN}^\dagger)$, where
$U_{\cN}:A\rightarrow BE$ is the unitary interaction associated with the channel. This
class of channels is a generalization of degradable channels~\cite{DS03},
and we will show that they enjoy the same property of additivity for coherent information.

\begin{proposition}
  \label{prop:additivity-info-degrad}
  Let quantum channels $\cN_1,\ldots, \cN_n$ be informationally degradable. Then they have
  additive coherent information:
  \begin{equation}
    \label{eq:prop-aid-1}
    I_c(\cN_1\otimes\ldots \otimes\cN_n)=\sum_{i=1}^n I_c(\cN_i).
  \end{equation}
  Especially, $Q(\cN)=I_c(\cN)$ for any informationally degradable channel $\cN$.
\end{proposition}

\begin{proof}
It suffices to show subadditivity. At first, we notice that due to the informational degradability,
for any $i$ it holds that
\[
  H(B_iV)-H(E_iV) \leq H(B_i)-H(E_i),
\]
where the entropies are evaluated on any state $\rho_{VB_iE_i}:=U_{\cN_i}\phi_{VA_i}U_{\cN_i}^\dagger$.
Using this, we can actually show uniform subadditivity,
\begin{equation*}\begin{split}
  I_c(\phi_{A_1\ldots A_n},\cN_1\otimes\ldots \otimes\cN_n)=&H(B_1\ldots B_n)-H(E_1\ldots E_n) \\
  =& H(B_1\ldots B_n)-H(E_1B_2\ldots B_n) \\
   &+H(E_1B_2\ldots B_n)-H(E_1E_2B_3\ldots B_n) \\
   &+\ldots                                     \\
   &+H(E_1\ldots E_{n-1}B_n)-H(E_1\ldots E_n)   \\
 \leq &\left(H(B_1)-H(E_1)\right)+\left(H(B_2)-H(E_2)\right)+\ldots + \left(H(B_n)-H(E_n)\right) \\
  =& \sum_{i=1}^n I_c(\phi_{A_i},\cN_i).
\end{split}\end{equation*}
\end{proof}
This implies Eq.~(\ref{eq:prop-aid-1}). The single-letter formula for quantum capacity follows
as a consequence directly.

\smallskip\noindent
\textbf{Remark}: We do not know whether product of informationally degradable channels is still informationally degradable.
This is why we prove additivity for multiple uses of channels, instead of proving this for any two channels.
A interesting problem we leave for future study is to find informationally degradable channels that are not degradable in the sense of Ref.~\cite{DS03}.

\section{Completely coherent information and quantum sum rate}

Given an isometry $U_{\cN}:A\rightarrow BE$, we say that the rate pair $(R_1,R_2)$ is an achievable joint quantum communication rate if for all $\epsilon>0$ there is an $n_0$ such that for all $n\geq n_0$ there are isometries $U_{\cE_n}: A_1A_2 \rightarrow A^nF$ and
decoders $U_{\cD_1}: B^n \rightarrow A_1^\prime D_1$ and $U_{\cD_2}: E^n \rightarrow A_2^\prime D_2$ with $\log |A_1| \geq nR_1 $ and $\log |A_2| \geq nR_2 $, and $A_1\cong A_1^\prime$ and $A_2\cong A_2^\prime$, such that
\begin{equation}\begin{split}\label{Eq:EncodeDecodeRateSum}
\rho_{V_1A_1 V_2A_2 D_1D_2F}:&=(U_{\cD_1}\otimes U_{\cD_2}) U_{\cN}^{\otimes n} U_{\cE}( \proj{\Phi_{V_1A_1}} \otimes \proj{\Phi_{V_2A_2}})U_{\cE}^\dg (U_{\cN}^{\otimes n})^\dg (U_{\cD_1}^\dg \otimes U_{\cD_2}^\dg) \\
 &\approx_\epsilon \proj{\Phi_{V_1A_1^\prime}} \otimes \proj{\Phi_{V_2A_2^\prime}}\otimes \proj{\sigma_{D_1D_2F}} = :\sigma_{V_1A_1^\prime V_2A_2^\prime D_1D_2F}
\end{split}\end{equation}
where
\begin{align}
\ket{\Phi_{V_1A_1}} & = \frac{1}{\sqrt{|A_1|}} \sum_{i = 1}^{|A_1|} \ket{i}_{V_1}\ket{i}_{A_1}\\
\ket{\Phi_{V_2A_2}} & = \frac{1}{\sqrt{|A_2|}} \sum_{i = 1}^{|A_2|} \ket{i}_{V_2}\ket{i}_{A_2},
\end{align}
$\rho \approx_\epsilon \sigma $ means $\| \rho -\sigma\|_1 \leq \epsilon$ and $ \ket{\sigma_{D_1D_2F}}$ is some pure state on $D_1D_2F$.

We define the \emph{joint quantum capacity} of an isometry $U_\cN: A \rightarrow BE$ to be
\begin{align}
Q_J(U_\cN) & = \max \left\{R_1+R_2| (R_1,R_2) \textrm{ is acheivable} \right\}.
\end{align}

Now, if $(R_1,R_2)$ is an achievable rate pair, for any $\epsilon$ we have encoders and decoders satisfying Eq.~(\ref{Eq:EncodeDecodeRateSum}). Thus, we have

\begin{align}
n(R_1 + R_2) + H(D_1)_\sigma-H(D_1F)_{\sigma} & \leq \log |A_1^\prime| + \log |A_2^\prime| + H(D_1)_\sigma-H(D_1F)_{\sigma} \nonumber\\
 & = H\left( A_1^\prime\right)_\sigma + H(V_2)_\sigma  + H(D_1)_\sigma-H(D_1F)_{\sigma}\nonumber\\
 & = H\left( A_1^\prime V_2 D_1\right)_\sigma-H(D_2)_{\sigma}\nonumber\\
 & = H\left( A_1^\prime V_2 D_1\right)_\sigma-H(A_2^\prime V_2 D_2)_{\sigma}\nonumber\\
& \approx H\left( A_1 V_2 D_1\right)_\rho-H(A_2 V_2 D_2)_{\rho}\nonumber\\
& = H(B^nV_2)_\mu-H(E^nV_2)_\mu \label{Eq:FirstState}
\end{align}
where
\begin{align}
\mu_{B^nE^nV_1V_2F} & = U_{\cN}^{\otimes n} U_{\cE}( \proj{\Phi_{V_1A_1}} \otimes \proj{\Phi_{V_2A_2}}) U_{\cE}^\dg (U_{\cN}^{\otimes n})^\dg.
\end{align}

We also have
\begin{align}
n(R_1 + R_2) + H(D_1F)_{\sigma}-H(D_1)_\sigma & \leq \log |A_1^\prime| + \log |A_2^\prime| + H(D_1F)_{\sigma}-H(D_1)_\sigma\nonumber\\
 & =  H\left( A_1^\prime\right)_\sigma + H(V_2)_\sigma  + H(D_1F)_{\sigma}-H(D_1)_\sigma\nonumber\\
& = H\left( A_1^\prime D_1V_2 F \right)_\sigma-H(D_2F)_{\sigma} \nonumber\\
& = H\left( A_1^\prime D_1V_2 F \right)_\sigma-H(A_2^\prime V_2 D_2F)_{\sigma}\nonumber\\
& \approx H(B^nV_2F)_\mu-H(E^nV_2F)_\mu.\label{Eq:SecondState}
\end{align}

Now, let
\begin{align}
\tilde{\mu}_{B^nE^nV_2FT} = \frac{1}{2}\mu_{B^nE^nV_2F}\otimes \proj{0}_T + \frac{1}{2}\mu_{B^nE^nV_2}\otimes \proj{0}_F \otimes \proj{1}_T.
\end{align}
This is a state that can be made with $n$ copies of $U_\cN$.  Taking the average of Eq.~(\ref{Eq:FirstState}) and Eq.~(\ref{Eq:SecondState}), and letting $W = V_2FT$ we find
\begin{align}
n(R_1 + R_2) & \lesssim \frac{1}{2} \left(  H(B^nV_2)_\mu-H(E^nV_2)_\mu +H(B^nV_2F)_\mu-H(E^nV_2F)_\mu \right)\\
& = H(B^nW)_{\tilde{\mu}} - H(E^nW)_{\tilde{\mu}}\\
& \leq I^{cc}(U_\cN^{\otimes n}) = n I^{cc}(U_\cN).
\end{align}
This implies $Q_J(U_{\cN}) \leq I^{cc}(U_\cN)$.

\bibliographystyle{apsrev}
\bibliography{shortq}

\begin{thebibliography}{22}
\expandafter\ifx\csname natexlab\endcsname\relax\def\natexlab#1{#1}\fi
\expandafter\ifx\csname bibnamefont\endcsname\relax
  \def\bibnamefont#1{#1}\fi
\expandafter\ifx\csname bibfnamefont\endcsname\relax
  \def\bibfnamefont#1{#1}\fi
\expandafter\ifx\csname citenamefont\endcsname\relax
  \def\citenamefont#1{#1}\fi
\expandafter\ifx\csname url\endcsname\relax
  \def\url#1{\texttt{#1}}\fi
\expandafter\ifx\csname urlprefix\endcsname\relax\def\urlprefix{URL }\fi
\providecommand{\bibinfo}[2]{#2}
\providecommand{\eprint}[2][]{\url{#2}}

\bibitem[{\citenamefont{Cover and Thomas}(1991)}]{CoverThomas}
\bibinfo{author}{\bibfnamefont{T.~M.} \bibnamefont{Cover}} \bibnamefont{and}
  \bibinfo{author}{\bibfnamefont{J.~A.} \bibnamefont{Thomas}},
  \emph{\bibinfo{title}{{E}lements of {I}nformation {T}heory}}
  (\bibinfo{publisher}{Wiley \&{} {S}ons}, \bibinfo{year}{1991}).

\bibitem[{\citenamefont{Di{V}incenzo et~al.}(1998)\citenamefont{Di{V}incenzo,
  Shor, and Smolin}}]{DSS98}
\bibinfo{author}{\bibfnamefont{D.}~\bibnamefont{Di{V}incenzo}},
  \bibinfo{author}{\bibfnamefont{P.~W.} \bibnamefont{Shor}}, \bibnamefont{and}
  \bibinfo{author}{\bibfnamefont{J.~A.} \bibnamefont{Smolin}},
  \bibinfo{journal}{Phys. Rev. A} \textbf{\bibinfo{volume}{57}},
  \bibinfo{pages}{830} (\bibinfo{year}{1998}),
  \bibinfo{note}{ar{X}iv:quant-ph/9706061}.

\bibitem[{\citenamefont{Smith and Smolin}(2007)}]{SS07}
\bibinfo{author}{\bibfnamefont{G.}~\bibnamefont{Smith}} \bibnamefont{and}
  \bibinfo{author}{\bibfnamefont{J.~A.} \bibnamefont{Smolin}},
  \bibinfo{journal}{Phys. Rev. Lett.} \textbf{\bibinfo{volume}{98}},
  \bibinfo{pages}{030501} (\bibinfo{year}{2007}).

\bibitem[{\citenamefont{Smith and Yard}(2008)}]{SY08}
\bibinfo{author}{\bibfnamefont{G.}~\bibnamefont{Smith}} \bibnamefont{and}
  \bibinfo{author}{\bibfnamefont{J.}~\bibnamefont{Yard}},
  \bibinfo{journal}{Science} \textbf{\bibinfo{volume}{321}},
  \bibinfo{pages}{1812} (\bibinfo{year}{2008}),
  \bibinfo{note}{arXiv:0807.4935}.

\bibitem[{\citenamefont{Hastings}(2009)}]{Hastings09}
\bibinfo{author}{\bibfnamefont{M.}~\bibnamefont{Hastings}},
  \bibinfo{journal}{Nat. Phys.} \textbf{\bibinfo{volume}{5}},
  \bibinfo{pages}{255} (\bibinfo{year}{2009}),
  \bibinfo{note}{ar{X}iv:0809.3972}.

\bibitem[{\citenamefont{Li et~al.}(2009)\citenamefont{Li, Winter, Zou, and
  Guo}}]{LWZG09}
\bibinfo{author}{\bibfnamefont{K.}~\bibnamefont{Li}},
  \bibinfo{author}{\bibfnamefont{A.}~\bibnamefont{Winter}},
  \bibinfo{author}{\bibfnamefont{X.}~\bibnamefont{Zou}}, \bibnamefont{and}
  \bibinfo{author}{\bibfnamefont{G.}~\bibnamefont{Guo}},
  \bibinfo{journal}{Phys. Rev. Lett.} \textbf{\bibinfo{volume}{103}},
  \bibinfo{pages}{120501} (\bibinfo{year}{2009}),
  \bibinfo{note}{ar{X}iv:0903.4308}.

\bibitem[{\citenamefont{Smith and Smolin}(2009)}]{SS09a}
\bibinfo{author}{\bibfnamefont{G.}~\bibnamefont{Smith}} \bibnamefont{and}
  \bibinfo{author}{\bibfnamefont{J.~A.} \bibnamefont{Smolin}},
  \bibinfo{journal}{Phys. Rev. Lett.} \textbf{\bibinfo{volume}{103}},
  \bibinfo{pages}{120503} (\bibinfo{year}{2009}),
  \bibinfo{note}{ar{X}iv:0904.4050}.

\bibitem[{\citenamefont{Cubitt et~al.}(2011)\citenamefont{Cubitt, Chen, and
  Harrow}}]{CCH11}
\bibinfo{author}{\bibfnamefont{T.~S.} \bibnamefont{Cubitt}},
  \bibinfo{author}{\bibfnamefont{J.}~\bibnamefont{Chen}}, \bibnamefont{and}
  \bibinfo{author}{\bibfnamefont{A.~W.} \bibnamefont{Harrow}},
  \bibinfo{journal}{Information Theory, IEEE Transactions on}
  \textbf{\bibinfo{volume}{57}}, \bibinfo{pages}{8114} (\bibinfo{year}{2011}).

\bibitem[{\citenamefont{Cubitt et~al.}(2015)\citenamefont{Cubitt, Elkouss,
  Matthews, Ozols, P{\'e}rez-Garc{\'\i}a, and Strelchuk}}]{cubitt15}
\bibinfo{author}{\bibfnamefont{T.}~\bibnamefont{Cubitt}},
  \bibinfo{author}{\bibfnamefont{D.}~\bibnamefont{Elkouss}},
  \bibinfo{author}{\bibfnamefont{W.}~\bibnamefont{Matthews}},
  \bibinfo{author}{\bibfnamefont{M.}~\bibnamefont{Ozols}},
  \bibinfo{author}{\bibfnamefont{D.}~\bibnamefont{P{\'e}rez-Garc{\'\i}a}},
  \bibnamefont{and}
  \bibinfo{author}{\bibfnamefont{S.}~\bibnamefont{Strelchuk}},
  \bibinfo{journal}{Nature communications} \textbf{\bibinfo{volume}{6}}
  (\bibinfo{year}{2015}).

\bibitem[{\citenamefont{Shannon}(1948)}]{Shannon48}
\bibinfo{author}{\bibfnamefont{C.~E.} \bibnamefont{Shannon}},
  \bibinfo{journal}{Bell Syst. Tech. J.} \textbf{\bibinfo{volume}{27}},
  \bibinfo{pages}{379} (\bibinfo{year}{1948}).

\bibitem[{\citenamefont{Bennett et~al.}(2002)\citenamefont{Bennett, Shor,
  Smolin, and Thapliyal}}]{BSST02}
\bibinfo{author}{\bibfnamefont{C.~H.} \bibnamefont{Bennett}},
  \bibinfo{author}{\bibfnamefont{P.~W.} \bibnamefont{Shor}},
  \bibinfo{author}{\bibfnamefont{J.~A.} \bibnamefont{Smolin}},
  \bibnamefont{and} \bibinfo{author}{\bibfnamefont{A.~V.}
  \bibnamefont{Thapliyal}}, \bibinfo{journal}{IEEE Trans. Inf. Theory}
  \textbf{\bibinfo{volume}{48}}, \bibinfo{pages}{2637} (\bibinfo{year}{2002}).

\bibitem[{\citenamefont{Smith et~al.}(2008)\citenamefont{Smith, Smolin, and
  Winter}}]{SSW06}
\bibinfo{author}{\bibfnamefont{G.}~\bibnamefont{Smith}},
  \bibinfo{author}{\bibfnamefont{J.}~\bibnamefont{Smolin}}, \bibnamefont{and}
  \bibinfo{author}{\bibfnamefont{A.}~\bibnamefont{Winter}},
  \bibinfo{journal}{IEEE Trans. Info. Theory} \textbf{\bibinfo{volume}{54}},
  \bibinfo{pages}{4208} (\bibinfo{year}{2008}),
  \bibinfo{note}{ar{X}iv:quant-ph/0607039}.

\bibitem[{\citenamefont{Lieb and Ruskai}(1973)}]{LiebRuskai}
\bibinfo{author}{\bibfnamefont{E.~H.} \bibnamefont{Lieb}} \bibnamefont{and}
  \bibinfo{author}{\bibfnamefont{M.-B.} \bibnamefont{Ruskai}},
  \bibinfo{journal}{J. Math. Phys.} \textbf{\bibinfo{volume}{14}},
  \bibinfo{pages}{1938} (\bibinfo{year}{1973}).

\bibitem[{\citenamefont{Yeung}(1997)}]{Yeung97}
\bibinfo{author}{\bibfnamefont{R.~W.} \bibnamefont{Yeung}},
  \bibinfo{journal}{Information Theory, IEEE Transactions on}
  \textbf{\bibinfo{volume}{43}}, \bibinfo{pages}{1924} (\bibinfo{year}{1997}).

\bibitem[{\citenamefont{Pippenger}(2003)}]{Pip03}
\bibinfo{author}{\bibfnamefont{N.}~\bibnamefont{Pippenger}},
  \bibinfo{journal}{Information Theory, IEEE Transactions on}
  \textbf{\bibinfo{volume}{49}}, \bibinfo{pages}{773} (\bibinfo{year}{2003}).

\bibitem[{\citenamefont{Devetak et~al.}(2006)\citenamefont{Devetak, Junge,
  King, and Ruskai}}]{DJKR06}
\bibinfo{author}{\bibfnamefont{I.}~\bibnamefont{Devetak}},
  \bibinfo{author}{\bibfnamefont{M.}~\bibnamefont{Junge}},
  \bibinfo{author}{\bibfnamefont{C.}~\bibnamefont{King}}, \bibnamefont{and}
  \bibinfo{author}{\bibfnamefont{M.~B.} \bibnamefont{Ruskai}},
  \bibinfo{journal}{Communications in mathematical physics}
  \textbf{\bibinfo{volume}{266}}, \bibinfo{pages}{37} (\bibinfo{year}{2006}).

\bibitem[{\citenamefont{Garc{\'\i}a-Patr{\'o}n
  et~al.}(2009)\citenamefont{Garc{\'\i}a-Patr{\'o}n, Pirandola, Lloyd, and
  Shapiro}}]{GPLS09}
\bibinfo{author}{\bibfnamefont{R.}~\bibnamefont{Garc{\'\i}a-Patr{\'o}n}},
  \bibinfo{author}{\bibfnamefont{S.}~\bibnamefont{Pirandola}},
  \bibinfo{author}{\bibfnamefont{S.}~\bibnamefont{Lloyd}}, \bibnamefont{and}
  \bibinfo{author}{\bibfnamefont{J.~H.} \bibnamefont{Shapiro}},
  \bibinfo{journal}{Physical review letters} \textbf{\bibinfo{volume}{102}},
  \bibinfo{pages}{210501} (\bibinfo{year}{2009}).

\bibitem[{\citenamefont{Oppenheim and Winter}(2005)}]{OW05}
\bibinfo{author}{\bibfnamefont{J.}~\bibnamefont{Oppenheim}} \bibnamefont{and}
  \bibinfo{author}{\bibfnamefont{A.}~\bibnamefont{Winter}},
  \bibinfo{journal}{arXiv preprint quant-ph/0511082}  (\bibinfo{year}{2005}).

\bibitem[{\citenamefont{Hayden et~al.}(2004)\citenamefont{Hayden, Jozsa, Petz,
  and Winter}}]{SSA-eq}
\bibinfo{author}{\bibfnamefont{P.}~\bibnamefont{Hayden}},
  \bibinfo{author}{\bibfnamefont{R.}~\bibnamefont{Jozsa}},
  \bibinfo{author}{\bibfnamefont{D.}~\bibnamefont{Petz}}, \bibnamefont{and}
  \bibinfo{author}{\bibfnamefont{A.}~\bibnamefont{Winter}},
  \bibinfo{journal}{Communications in mathematical physics}
  \textbf{\bibinfo{volume}{246}}, \bibinfo{pages}{359} (\bibinfo{year}{2004}).

\bibitem[{\citenamefont{Petz}(1986)}]{Petz1}
\bibinfo{author}{\bibfnamefont{D.}~\bibnamefont{Petz}},
  \bibinfo{journal}{Communications in mathematical physics}
  \textbf{\bibinfo{volume}{105}}, \bibinfo{pages}{123} (\bibinfo{year}{1986}).

\bibitem[{\citenamefont{Petz}(1988)}]{Petz2}
\bibinfo{author}{\bibfnamefont{D.}~\bibnamefont{Petz}}, \bibinfo{journal}{The
  Quarterly Journal of Mathematics} \textbf{\bibinfo{volume}{39}},
  \bibinfo{pages}{97} (\bibinfo{year}{1988}).

\bibitem[{\citenamefont{Devetak and Shor}(2005)}]{DS03}
\bibinfo{author}{\bibfnamefont{I.}~\bibnamefont{Devetak}} \bibnamefont{and}
  \bibinfo{author}{\bibfnamefont{P.~W.} \bibnamefont{Shor}},
  \bibinfo{journal}{Comm. Math. Phys.} \textbf{\bibinfo{volume}{256}},
  \bibinfo{pages}{287} (\bibinfo{year}{2005}),
  \bibinfo{note}{ar{X}iv:quant-ph/0311131}.

\end{thebibliography}

\end{document}